\documentclass[titlepage,12pt]{article}
\usepackage{amssymb}
\usepackage{amsfonts}
\usepackage{amsmath}
\usepackage{natbib}
\usepackage{graphicx}
\usepackage{rotating}
\usepackage{setspace}
\usepackage{url}
\usepackage[margin=1.25in]{geometry}

\newcommand{\var}{\operatorname{var}}
\newcommand{\E}{\operatorname{E}}
\newcommand{\atanh}{\operatorname{atanh}}
\newcommand{\ESS}{\textit{ESS}}
\newcommand{\RSS}{\textit{RSS}}
\newcommand{\NSE}{\textit{NSE}}
\newcommand{\RNE}{\textit{RNE}}
\newcommand{\egarch}[1]{\texttt{egarch\_#1}}
\newcommand{\Rbar}{\overline R}
\newcommand{\Rmax}{R_\text{max}}
\newcommand{\Vstrut}[2][0pt]{\rule[#1]{0pt}{#2}}
\newcommand{\vhat}{\mathrm{v}\widehat{\mathrm{a}}\mathrm{r}}

\newtheorem{algorithm}{Algorithm}
\newtheorem{condition}{Condition}
\newtheorem{proposition}{Proposition}
\newenvironment{proof}[1][Proof]{\noindent\textbf{#1.} }{\ \rule{0.5em}{0.5em}}

\begin{document}

\title{Adaptive Sequential Posterior Simulators for Massively Parallel
Computing Environments}
\author{Garland Durham\thanks{Quantos Analytics, LLC; \texttt{garland@quantosanalytics.org.}}
and John Geweke\thanks{University of Technology Sydney (Australia), Erasmus University (The
Netherlands) and University of Colorado (USA), \texttt{John.Geweke@uts.edu.au}. 
Geweke acknowledges partial financial support from Australian Research
Council grants DP110104372 and DP130103356.}\textsuperscript{,}\thanks{%
We acknowledge useful comments from Nicolas Chopin, discussions with Ron
Gallant, and tutorials in CUDA programming from Rob Richmond. We bear sole
responsibility for the content of the paper. An earlier version of this
work was posted with the title \textquotedblleft Massively Parallel
Sequential Monte Carlo for Bayesian Inference.\textquotedblright}}
\date{April 6, 2013}
\maketitle

\begin{abstract} 
Massively parallel desktop computing capabilities now well within the reach of
individual academics modify the environment for posterior simulation in
fundamental and potentially quite advantageous ways.  
But to fully exploit these benefits algorithms that conform to parallel
computing environments are needed. 
Sequential Monte Carlo comes very close to this ideal whereas other approaches
like Markov chain Monte Carlo do not. 
This paper presents a sequential posterior simulator well suited to this
computing environment.  
The simulator makes fewer analytical and programming demands on investigators,
and is faster, more reliable and more complete than conventional posterior
simulators. 
The paper extends existing sequential Monte Carlo methods and theory to provide
a thorough and practical foundation for sequential posterior simulation that is
well suited to massively parallel computing environments.  
It provides detailed recommendations on implementation, yielding an algorithm
that requires only code for simulation from the prior and evaluation of prior
and data densities and works well in a variety of applications representative
of serious empirical work in economics and finance.  
The algorithm is robust to pathological posterior distributions, generates
accurate marginal likelihood approximations, and provides estimates of
numerical standard error and relative numerical efficiency intrinsically. 
The paper concludes with an application that illustrates the potential of these
simulators for applied Bayesian inference.



\end{abstract}

\vfill


\section{Introduction}
\label{sec:introduction}

Bayesian approaches have inherent advantages in solving inference and decision
problems, but practical applications pose challenges for computation. As these
challenges have been met Bayesian approaches have proliferated and contributed
to the solution of applied problems. McGrayne (2011) has recently conveyed
these facts to a wide audience.

The evolution of Bayesian computation over the past half-century has conformed
with exponential increases in speed and decreases in the cost of computing. The
influence of computational considerations on algorithms, models, and the way
that substantive problems are formulated for statistical inference can be
subtle but is hard to over-state. Successful and innovative basic and applied
research recognizes the characteristics of the tools of implementation from the
outset and tailors approaches to those tools.

Recent innovations in hardware (graphics processing units, or GPUs)
provide individual investigators with massively parallel desktop processing at
reasonable cost. Corresponding developments in software (extensions of the C
programming language and mathematical applications software) make these
attractive platforms for scientific computing. This paper extends and applies
existing sequential Monte Carlo (SMC) methods for posterior simulation in this
context.  The extensions fill gaps in the existing theory to provide a thorough
and practical foundation for sequential posterior simulation, using approaches
suited to massively parallel computing environments. The application produces
generic posterior simulators that make substantially fewer analytical and
programming demands on investigators implementing new models and provides
faster, more reliable and more complete posterior simulation than do existing
methods inspired by conventional predominantly serial computational methods. 

Sequential posterior simulation grows out of SMC methods developed over the
past 20 years and applied primarily to state space models (``particle
filters"). Seminal contributions include Baker (1985, 1987), Gordon et
al.~(1993), Kong et al.~(1994), Liu and Chen (1995, 1998), Chopin (2002, 2004),
Del Moral et al.~(2006), Andrieu et al.~(2010), Chopin and Jacob (2010) and Del
Moral et al.~(2011).

The posterior simulator proposed in this paper, which builds largely on Chopin
(2002, 2004), has attractive properties.
\begin{itemize}
\item It is highly generic (easily adaptable to new models). All that is required
    of the user is code to generate from the prior and evaluate the prior and
    data densities.
\item It is computationally efficient relative to available alternatives.  This
    is in large part due to the specifics of our implementation of the
        simulator, which makes effective use of low-cost massively parallel
        computing hardware. 
\item Since the simulator provides a sample from the posterior density at each
    observation date conditional on information available at that time, it is
    straightforward to compute moments of arbitrary functions of interest
    conditioning on relevant information sets.  Marginal likelihood and
    predictive scores, which are key elements of Bayesian analysis, are
    immediately available.  Also immediately available is the probability
    integral transform at each observation date, which provides a powerful
    diagnostic tool for model exploration (``generalized residuals"; e.g.,
    Diebold et al., 1998). 
\item Estimates of numerical standard error and relative numerical efficiency are
    provided as an intrinsic part of the simulator.  This is important but
    often neglected information for the researcher interested in generating
    dependable and reliable results.
\item The simulator is robust to irregular posteriors (e.g., multimodality), as
    has been pointed out by Jasra et al.~(2007) and others.
\item The simulator has well-founded convergence properties.
\end{itemize}

But, although the basic idea of sequential posterior simulation goes back at
least to Chopin (2002, 2004), and despite its considerable appeal, the idea has
seen essentially no penetration in mainstream applied econometrics.
Applications have been limited to relatively simple illustrative examples
(although this has begun to change very recently; see Herbst and Schorfheide,
2012; Fulop and Li, 2012; Chopin et al., 2012).  This is in stark contrast to
applications of SMC methods to state space filtering (``particle filters"),
which have seen widespread use.

Relative to Markov chain Monte Carlo (MCMC), which has become a mainstay of
applied work, sequential posterior simulators are computationally costly when
applied in a conventional serial computing environment.  Our interest in
sequential posterior simulation is largely motivated by the recent availability
of low cost hardware supporting massively parallel computation. SMC methods are
much better suited to this environment than is MCMC.   

The massively parallel hardware device used in this paper is a commodity
graphical processing unit (GPU), which provides hundreds of cores at a cost of
well under one dollar (US) each.  But in order to realize the huge potential
gains in computing speed made possible by such hardware, algorithms that
conform to the single-instruction multiple-data (SIMD) paradigm are needed.
The simulators presented in this paper conform to the SIMD paradigm by design
and realize the attendant increases in computing speed in practical
application.

But there are also some central issues regarding properties of sequential
posterior simulators that have not been resolved in the literature.  Since
reliable applied work depends on the existence of solid theoretical
underpinnings, we address these as well.  

Our main contributions are as follows.
\begin{enumerate}
    \item \emph{Theoretical basis.} Whereas the theory for sequential posterior
        simulation as originally formulated by Chopin (2002, 2004) assumes that
        key elements of the algorithm are fixed and known in advance, practical
        applications demand algorithms that are adaptive, with these elements
        constructed based on the information provided by particles generated in
        the course of running the algorithm.

        While some progress has been made in developing a theory that applies
        to such adaptive algorithms, (e.g., Douc and Moulines, 2008), we expect
        that a solid theory that applies to the kind of algorithms that are
        used and needed in practice is likely to remain unattainable in the
        near future.

        In Section \ref{sec:algorithms} we provide an approach that addresses
        this problem in a different way, providing a posterior simulator that
        is highly adapted to the models and data at hand, while satisfying the
        relatively straightforward conditions elaborated in the original work
        of Chopin (2002, 2004).
    \item \emph{Numerical accuracy.} While Chopin (2004) provides a critical
        central limit theorem, the extant literature does not provide any means
        of estimating the variance, which is essential for assessing the
        numerical accuracy and relative numerical efficiency of moment
        approximations (Geweke, 1989).  This problem has proved difficult in
        the case of MCMC (Flegal and Jones, 2010) and appears to have been
        largely ignored in the SMC literature.  In Section
        \ref{sec:parpostsims} we propose an approach to resolving this issue
        which entails no additional computational effort and is natural in the
        context of the massively parallel environment as well as key in making
        efficient use of it.  The idea relies critically upon the theoretical
        contribution noted above. 
    \item \emph{Marginal likelihood.} While the approach to assessing the
        asymptotic variance of moment approximations noted in the previous
        point is useful, it does not apply directly to marginal likelihood,
        which is a critical element of Bayesian analysis.  We address this
        issue in Section \ref{sec:PML}.
    \item \emph{Parallel implementation.}  It has been noted that sequential
        posterior simulation is highly amenable to parallelization going back
        to at least Chopin (2002), and there has been some work toward
        exploring parallel implementations (e.g., Lee et al., 2010; Fulop and
        Li, 2012).  Building on this work, we provide a software package that
        includes a full GPU implementation of the simulator, with specific
        details of the algorithm tuned to the massively parallel environment as
        outlined in Section \ref{sec:algorithms} and elsewhere in this paper.
        The software is fully object-oriented, highly modular, and easily
        extensible.  New models are easily added, providing the full benefits
        of GPU computing with little effort on the part of the user.  Geweke et
        al.~(2013) provides an example of such an implementation for the logit
        model.  In future work, we intend to build a library of models and
        illustrative applications utilizing this framework which will be freely
        available.  Section \ref{ss:software} provides details about this software.

    \item \emph{Specific recommendations.} The real test of the simulator is in
        its application to problems that are characteristic of the scale and
        complexity of serious disciplinary work.  In Section
        \ref{sec:application}, we provide one such application to illustrate
        key ideas.  In other ongoing work, including Geweke et al.~(2013) and a
        longer working paper written in the process of this research (Durham
        and Geweke, 2011), we provide several additional substantive
        applications.  As part of this work, we have come up with some specific
        recommendations for aspects of the algorithm that we have found to work
        well in a wide range of practical applications.  These are described in
        Section \ref{ss:details} and implemented fully in the software package
        we are making available with this paper.
\end{enumerate}

\section{Posterior simulation in a massively parallel computing environment} 
\label{sec:environment}

The sequential simulator proposed in this paper is based on ideas that go back
to Chopin (2002, 2004), with even earlier antecedents including Gilkes and
Berzuini (2001),  and Fearnhead (1998).  The simulator begins with a sample of
parameter vectors (``particles") from the prior distribution.  Data is
introduced in batches, with the steps involved in processing a single batch of
data referred to as a cycle.  At the beginning of each cycle, the particles
represent a sample from the posterior conditional on information available up
to that observation date.  As data is introduced sequentially, the posterior is
updated using importance sampling (Kloek and van Dijk, 1978), the appropriately
weighted particles representing an importance sample from the posterior at each
step.  As more data is introduced, the importance weights tend to become
``unbalanced" (a few particles have most of the weight, while many others have
little weight), and importance sampling becomes increasingly inefficient.  When
some threshold is reached, importance sampling stops and the cycle comes to an
end.  At this point, a resampling step is undertaken, wherein particles are
independently resampled in proportion to their importance weights.  After this
step, there will be many copies of particles with high weights, while particles
with low weights will tend to drop out of the sample.  Finally, a sequence of
Metropolis steps is undertaken in order to rejuvenate the diversity of the
particle sample.  At the conclusion of the cycle, the collection of particles
once again represents a sample from the posterior, now incorporating the
information accrued from the data newly introduced. 

This sequential introduction of data is natural in a time-series setting, but
also applicable to cross-sectional data.  In the latter case, the sequential
ordering of the data is arbitrary, though some orderings may be more useful than
others.

Much of the appeal of this simulator is due to its ammenability to
implementation using massively parallel hardware.  Each particle can be handled
in a distinct thread, with all threads updated concurrently at each step.
Communication costs are low.  In applications where computational cost is an
important factor, nearly all of the cost is incurred in evaluating densities of
data conditional on a candidate parameter vector and thus isolated to
individual threads.  Communication costs in this case are a nearly negligible
fraction of the total computational burden.

The key to efficient utilization of the massively parallel hardware is that the
workload be divided over many SIMD threads.  For the GPU hardware used in this
paper optimal usage involves tens of thousands of threads.  In our
implementation, particles are organized in a relatively small number of groups
each with a relatively large number of particles (in the application in Section
\ref{sec:application} there are $2^6$ groups of $2^{10}$ particles each).  This
organization of particles in groups is fundamental to the supporting theory.
Estimates of numerical error and relative numerical efficiency are generated as
an intrinsic part of the algorithm with no additional computational cost, while
the reliability of these estimates is supported by a sound theoretical
foundation.  This structure is natural in a massively parallel computing
environment, as well critical in making efficient use of it.

The remainder of this section provides a more detailed discussion of key issues
involved in posterior simulation in a massively parallel environment.  It
begins in Section \ref{subsec:compenv} with a discussion of the relevant
features of the hardware and software used.  Section
\ref{subsec:notation_assumption} sets up a generic model for Bayesian inference
along with conditions used in deriving the analytical properties of various
sequential posterior simulators in Section \ref{sec:algorithms}.  Section
\ref{sec:parpostsims} stipulates a key convergence condition for posterior
simulators, and then shows that if this condition is met there are attractive
generic methods for approximating the standard error of numerical approximation
in massively parallel computing environments.  Section \ref{sec:algorithms}
then develops sequential posterior simulators that satisfy this condition.

\subsection{Computing environment}
\label{subsec:compenv}

The particular device that motivates this work is the graphics processing unit
(GPU). As a practical matter several GPU's can be incorporated in a single
server (the ``host") with no significant complications, and desktop computers
that can accommodate up to eight GPU's are readily available. The single- and
multiple-GPU environments are equivalent for our purposes. A single GPU
consists of several multiprocessors, each with several cores. The GPU has
global memory shared by its multiprocessors, typically one to several gigabytes
(GB) in size, and local memory specific to each multiprocessor, typically on
the order of 50 to 100 kilobytes (KB) per multiprocessor. (For example, this
research uses a single Nvidia GTX 570 GPU with 15 multiprocessors, each with 32
cores. The GPU has 1.36 GB of local memory, and each multiprocessor has 49 KB
of memory and 32 KB of registers shared by its cores.) The bus that transfers
data between GPU global and local memory is significantly faster than the bus
that transfers data between host and device, and accessing local memory on the
multiprocessor is faster yet.  For greater technical detail on GPU's, see
Hendeby et al.~(2010), Lee et al.~(2010) and Souchard et al.~(2010).

This hardware has become attractive for scientific computing with the extension
of scientific programming languages to allocate the execution of instructions
between host and device and facilitate data transfer between them. Of these the
most significant has been the compute unified device architecture (CUDA)
extension of the C programming language (Nvidia, 2013). CUDA abstracts the
host-device communication in a way that is convenient to the programmer yet
faithful to the aspects of the hardware important for writing efficient code.

Code executed on the device is contained in special functions called kernels
that are invoked by the host code. Specific CUDA instructions move the data on
which the code operates from host to device memory and instruct the device to
organize the execution of the code into a certain number of blocks with a
certain number of threads each.  The allocation into blocks and threads is the
virtual analogue of the organization of a GPU into multiprocessors and cores.

While the most flexible way to develop applications that make use of GPU
parallelization is through C/C++ code with direct calls to the vendor-supplied
interface functions, it is also possible to work at a higher level of
abstraction. For example, a growing number of mathematical libraries have been
ported to GPU hardware (e.g., Innovative Computing Laboratory, 2013). Such
libraries are easily called from standard scientific programming languages and
can yield substantial increases in performance for some applications. In
addition, Matlab (2013) provides a library of kernels, interfaces for calling
user-written kernels, and functions for host-device data transfer from within
the Matlab workspace.

\subsection{Models and conditions}
\label{subsec:notation_assumption}

We augment standard notation for data, parameters and models. The relevant
observable random vectors are $Y_{t}$ $\left( t=1,\ldots ,T\right) $ and 
$Y_{t_{1}:t_{2}}$ denotes the collection $\left\{ Y_{t_{1}},\ldots
,Y_{t_{2}}\right\} $. The observation of $Y_{t}$ is $y_{t}$, 
$y_{t_{1}:t_{2}} $ denotes the collection $\left\{ y_{t_{1}},\ldots
,y_{t_{2}}\right\} $, and therefore $y_{1:T}$ denotes the data. This
notation assumes ordered observations, which is natural for time series. If 
$\left\{ Y_{t}\right\} $ is independent and identically distributed the
ordering is arbitrary.

A model for Bayesian inference specifies a $k\times 1$ unobservable
parameter vector $\theta \in \Theta $ and a conditional density 
\begin{equation}
p\left( Y_{1:T}\mid \theta \right) =\mathop{\displaystyle \prod }
\limits_{t=1}^{T}p\left( Y_{t}\mid Y_{1:t-1},\theta \right)
\label{p(y|theta)}
\end{equation}
with respect to an appropriate measure for $Y_{1:T}$. The model also
specifies a prior density $p\left(\theta\right) $ with respect to a measure 
$\nu$ on $\Theta$. The posterior density $p\left( \theta \mid y_{1:T}\right) $ follows in
the usual way from \eqref{p(y|theta)} and $p\left(\theta\right)$.

The objects of Bayesian inference can often be written as posterior moments
of the form $E\left[ g\left( \theta \right) \mid y_{1:T}\right] $, and going
forward we use $g\left( \theta \right) $ to refer to such a generic function
of interest. \ Evaluation of $g\left( \theta \right) $ may require
simulation, e.g.~$g\left( \theta \right) =E\left[ h\left( Y_{T+1:T+f}\right)
\mid \theta \right] $, and conventional approaches based on the posterior
simulation sample of parameters (e.g.~Geweke, 2005, Section 1.4) apply in
this case. The marginal likelihood $p\left( y_{1:T}\right) $ famously does
not take this form, and Section \ref{sec:PML} takes up its approximation
with sequential posterior simulators and, more specifically, in the context of
the methods developed in this paper.

Several conditions come into play in the balance of the paper.

\begin{condition}
\label{cond:prior_evaluate}(Prior distribution). The model specifies a
proper prior distribution. The prior density kernel can be evaluated with
SIMD-compatible code. Simulation from the prior distribution must be
practical but need not be SIMD-compatible.
\end{condition}

It is well understood that a model must take a stance on the distribution of
outcomes $Y_{1:T}$ \textit{a priori} if it is to have any standing in formal
Bayesian model comparison, and that this requires a proper prior distribution.
This requirement is fundamentally related to the generic structure of
sequential posterior simulators, including those developed in Section
\ref{sec:algorithms}, because they require a distribution of $\theta $ before
the first observation is introduced. Given a proper prior distribution the
evaluation and simulation conditions are weak.  Simulation from the prior
typically involves minimal computational cost and thus we do not require it to
be SIMD-compatible.  However, in practice it will often be so.

\begin{condition}
\label{cond:LF_evaluate}(Likelihood function evaluation) The sequence of
conditional densities 
\begin{equation}
p\left( y_{t}\mid y_{1:t-1},\theta \right)\quad\left( t=1,\ldots
,T\right)
\label{model_conditional_pdf}
\end{equation}
can be evaluated with SIMD-compatible code for all $\theta \in \Theta $.
\end{condition}

Evaluation with SIMD-compatible code is important to computational efficiency
because evaluation of (\ref{model_conditional_pdf}) constitutes almost all of
the floating point operations in typical applications for the algorithms
developed in this paper.  Condition \ref{cond:LF_evaluate} excludes situations
in which unbiased simulation of (\ref{model_conditional_pdf}) is possible but
exact evaluation is not, which is often the case in nonlinear state space
models, samples with missing data, and in general any case in which a closed
form analytical expression for (\ref{model_conditional_pdf}) is not available.
A subsequent paper will take up this extension.

\begin{condition}
    \label{cond:bounded_like}
    (Bounded likelihood) The data density $p(y_{1:T}|\theta)$ 
    is bounded above by $\overline{p}<\infty $ for all 
    $\theta\in\Theta$.
\end{condition}

This is one of two sufficient conditions for the central limit theorem
invoked in Section \ref{subsec:nonadaptive}. It is commonly but not
universally satisfied. When it is not, it can often be attained by minor and
innocuous modification of the likelihood function; Section 
\ref{sec:application} provides an example.

\begin{condition}
\label{cond:prior_var}(Existence of prior moments) If the algorithm is used
to approximate $\mathrm{E}\left[ g\left( \theta \right) \mid y_{1:T}\right]$, 
then $\mathrm{E}\left[ g\left( \theta \right) ^{2+\delta}\right]<\infty$
for some $\delta>0$.
\end{condition}

In any careful implementation of posterior simulation the existence of
relevant posterior moments must be verified analytically. This condition,
together with Condition \ref{cond:bounded_like}, is sufficient for the
existence of $\E\left[ g\left( \theta \right) \mid y_{1:T}\right] $
and $\var\left[ g\left( \theta \right) \mid y_{1:T}\right] $.
Condition \ref{cond:prior_var} also comes into play in establishing a
central limit theorem.

\subsection{Assessing numerical accuracy and relative numerical efficiency}
\label{sec:parpostsims}

Consider the implementation of any posterior simulator in a parallel
computing environment like the one described in Section \ref{subsec:compenv}. 
Following the conventional approach, we focus on the posterior simulation
approximation of 
\begin{equation}
\overline{g}=\E\left[ g\left( \theta \right) \mid y_{1:T}\right].
\label{post_moment}
\end{equation}
The posterior simulator operates on parameter vectors, or particles, 
$\theta _{jn}$ organized in $J$ groups of $N$ particles each; define $\mathcal{J}
=\left\{ 1,\ldots ,J\right\} $ and $\mathcal{N}=\left\{ 1,\ldots ,N\right\} $. 
This organization is fundamental to the rest of the paper.

Each particle is associated with a core of the device. The CUDA software
described in Section \ref{subsec:compenv} copes with the case in which $JN$
exceeds the number of cores in an efficient and transparent fashion. Ideally
the posterior simulator is SIMD-compatible, with identical instructions
executed on all particles. Some algorithms, like the sequential posterior
simulators developed in Section \ref{sec:algorithms}, come quite close to
this ideal. Others, for example Gibbs samplers with Metropolis steps, may
not. The theory in this section applies regardless. However, the advantages
of implementing a posterior simulator in a parallel computing environment
are driven by the extent to which the algorithm is SIMD-compatible.

Denote the evaluation of the function of interest at particle $\theta _{jn}$
by $g_{jn}=g\left( \theta _{jn}\right) $ and within-group means by 
$\overline{g}_{j}^{N}=N^{-1}\sum_{n=1}^{N}g_{jn}$ $\left( j\in \mathcal{J}
\right) $. The posterior simulation approximation of $\overline{g}$ is the
grand mean 
\begin{equation}
\overline{g}^{\left( J,N\right) }=J^{-1}\sum_{j=1}^{J}
\overline{g}_{j}^{N}=\left( JN\right) ^{-1}\sum_{j=1}^{J}\sum_{n=1}^{N}g_{jn}.
\label{garnd_approx}
\end{equation}

In general we seek posterior simulators with the following properties.

\begin{condition}
\label{cond:normal_approx}(Asymptotic normality of posterior moment
approximation) The random variables $\overline{g}_{j}^{N}\,\left( j\in 
\mathcal{J}\right) $ are independently and identically distributed. There
exists $v>0$ for which
\begin{equation}
N^{1/2}\left( \overline{g}_{j}^{N}-\overline{g}\right) \overset{d}
{\rightarrow }N\left( 0,v\right) \text{ \ }\left( j\in \mathcal{J}\right) 
\label{CLT copy(1)}
\end{equation}
as $N\rightarrow\infty$.
\end{condition}            

Going forward, $v$ will denote a generic variance in a central limit theorem.
Convergence \eqref{CLT copy(1)} is part of the rigorous foundation of posterior
simulators, e.g. Geweke (1989) for importance sampling, Tierney (1994) for
MCMC, and Chopin (2004) for a sequential posterior simulator.  For importance
sampling $\overline{g}_{j}^{N}=\sum_{n=1}^{N}\omega \left( \theta _{jn}\right)
g_{jn}/\sum_{n=1}^{N}\omega \left( \theta _{jn}\right) $ $\left( j\in
\mathcal{J}\right) $, where $\omega \left( \theta \right) $ is the ratio of the
posterior density kernel to the source density kernel.

For \eqref{CLT copy(1)} to be of any practical use in assessing numerical
accuracy there must also be a simulation-consistent approximation of $v$.  Such
approximations are immediate for importance sampling. They have proven more
elusive for MCMC; e.g., see Flegal and Jones (2010) for discussion and an
important contribution. To our knowledge there is no demonstrated
simulation-consistent approximation of $v$ for sequential posterior simulators
in the existing literature. 

However, for any simulator satisfying Condition
\ref{cond:normal_approx} this is straightforward.
Given Condition \ref{cond:normal_approx}, it is immediate that
\begin{equation}
\left( JN\right) ^{1/2}\left( \overline{g}^{\left( J,N\right) }-\overline{g}\right) 
\overset{d}{\rightarrow }N\left( 0,v\right).
\label{CLT}
\end{equation}
Define the estimated posterior simulation variance
\begin{equation}
\widehat{v}^{\left( J,N\right) }\left( g\right) =\left[ N/\left( J-1\right) %
\right] \sum_{j=1}^{J}\left( \overline{g}_{j}^{N}-\overline{g}^{\left(
J,N\right) }\right) ^{2}\label{vhat}
\end{equation}
and the numerical standard error (\NSE)
\begin{equation}
\NSE^{\left( J,N\right) }\left( g\right) =\left[ \widehat{v}^{\left( J,N\right)
}/JN\right] ^{1/2}=\left\{ \left[ J\left( J-1\right) \right]
^{-1}\sum_{j=1}^{J}\left( \overline{g}_{j}^{N}-\overline{g}^{\left(
J,N\right) }\right)^2 \right\} ^{1/2}.
\label{NSE_def}
\end{equation}
Note the different scaling conventions in \eqref{vhat} and \eqref{NSE_def}: in
\eqref{vhat} the scaling is selected so that $\widehat{v}^{\left( J,N\right)
}\left( g\right)$ approximates $v$ in \eqref{CLT} because we will use this
expression mainly for methodological developments; in \eqref{NSE_def} the
scaling is selected to make it easy to appraise the reliability of numerical
approximations of posterior moments like those reported in Section
\ref{sec:application}.

One conventional assessment of the efficiency of a posterior simulator is
its relative numerical efficiency (\RNE) (Geweke, 1989),
\begin{equation*}
\RNE^{\left(J,N\right)}(g)=\mathrm{v}\widehat{\mathrm{a}}\mathrm{r}^{\left(
J,N\right) }\left( g\right) /\widehat{v}^{\left( J,N\right) }\left( g\right)
\end{equation*}
where $\mathrm{v}\widehat{\mathrm{a}}\mathrm{r}^{\left( J,N\right) }\left(
g\right) =\left( NJ-1\right) ^{-1}\sum_{j=1}^{J}\sum_{n=1}^{N}\left( g_{jn}-
\overline{g}^{\left( J,N\right) }\right) ^{2}$, the simulation approximation 
of $\mathrm{var}\left[ g\left( \theta \right) \mid y_{1:T}\right] $.

\begin{proposition}
\label{prop:CLT_var_approx}Condition \ref{cond:normal_approx} implies
\begin{equation}
\left( J-1\right) \widehat{v}^{\left( J,N\right) }\left( g\right) /v
\overset{d}{\rightarrow }\chi ^{2}\left( J-1\right)
\label{p1A}
\end{equation}
and
\begin{equation}
\left( JN\right) ^{1/2}\left( \overline{g}^{\left( J,N\right) }-\overline{g}
\right) /\left[ \widehat{v}^{\left( J,N\right) }\left( g\right) \right]
^{1/2}\overset{d}{\rightarrow }t\left( J-1\right),
\label{p1B}
\end{equation}
both as $N\rightarrow \infty $. Conditions \ref{cond:prior_evaluate} through 
\ref{cond:prior_var} imply
\begin{equation}
E\left[ \widehat{v}^{\left( J,N\right) }\left( g\right) \right] =var\left[ 
\overline{g}^{\left( J,N\right) }\right].
\label{p1C}
\end{equation}
\end{proposition}

\begin{proof}
From Condition \ref{cond:normal_approx},
\begin{equation}
N\sum_{j=1}^{J}\left( \overline{g}_{j}^{N}-\overline{g}^{\left( J,N\right)}\right) ^{2}/v
\overset{d}{\rightarrow }\chi ^{2}\left( J-1\right)
\label{chi2limit}
\end{equation}
as $N\rightarrow \infty $. Substituting \eqref{vhat} in \eqref{chi2limit}
yields \eqref{p1A}. Condition \ref{cond:normal_approx} also implies
\begin{equation}
\left( \frac{NJ}{v}\right) ^{1/2}\left( \overline{g}^{\left( J,N\right) }-
\overline{g}\right) \overset{d}{\rightarrow }N\left( 0,1\right).
\label{normlimit}
\end{equation}
as $N\rightarrow \infty $. Since \eqref{chi2limit} and \eqref{normlimit} are
independent in the limiting distribution, we have \eqref{p1B}.

Conditions \ref{cond:prior_evaluate} through \ref{cond:prior_var} imply the
existence of the first two moments of $g_{jn}$. Then \eqref{p1C} follows
from the fact that the approximations $\overline{g}_{j}$ are independent and
identically distributed across $j=1,\ldots ,J$.
\end{proof}

\section{Parallel sequential posterior simulators} 
\label{sec:algorithms}

We seek a posterior simulator that is generic, requiring little or no
intervention by the investigator in adapting it to new models beyond the
provision of the software implicit in Conditions \ref{cond:prior_evaluate} and
\ref{cond:LF_evaluate}.  It should reliably assess the numerical accuracy of
posterior moment approximations, and this should be achieved in substantially
less execution time than would be required using the same or an alternative
posterior simulator for the same model in a serial computing environment.

Such simulators are necessarily adaptive: they must use the features of the
evolving posterior distribution, revealed in the particles $\theta _{jn}\left(
j\in \mathcal{J}\!,\,n\in\mathcal{N}\right) $, to design the next steps in the
algorithm. This practical requirement has presented a methodological conundrum
in the sequential Monte Carlo literature, because the mathematical
complications introduced by even mild adaptation lead to analytically
intractable situations in which demonstration of the central limit theorem in
Condition \ref{cond:normal_approx} is precluded.  This section solves this
problem and sets forth a particular adaptive sequential posterior simulator
that has been successful in applications we have undertaken with a wide variety
of models.  Section \ref{sec:application} details one such application.

Section \ref{subsec:nonadaptive} places a nonadaptive sequential posterior
simulator (Chopin, 2004) that satisfies Condition \ref{cond:normal_approx} into
the context developed in the previous section. Section \ref{subsec:adaptive}
introduces a technique that overcomes the analytical difficulties associated
with adaptive simulators. It is generic, simple and  imposes little additional
overhead in a parallel computing environment.
Section \ref{ss:details} provides details on a particular variant of the
algorithm that we have found to be successful for a variety of models
representative of empirical research frontiers in economics and finance.
Section \ref{ss:software} discusses a software package implementing the
algorithm that we are making available.

\subsection{Nonadaptive simulators}
\label{subsec:nonadaptive}

We rely on the following mild generalization of the sequential Monte Carlo
algorithm of Chopin (2004), cast in the parallel computing environment detailed
in the previous section. In this algorithm, the observation dates at which each
cycle terminates ($t_1,\dots,t_L$) and the parameters involved in specifying
the Metropolis updates ($\lambda_1,\dots,\lambda_L$) are assumed to be fixed
and known in advance, in conformance with the conditions specified by Chopin
(2004).

\begin{algorithm}
\label{alg:general_nonadaptive}(Nonadaptive) \ Let $t_{0},\ldots
,t_{L}$ be fixed integers with $0=t_{0}<t_{1}<\ldots <t_{L}=T$ and let 
$\lambda _{1},\ldots ,\lambda _{L}$ be fixed vectors.
\end{algorithm}

\begin{enumerate}
\item Initialize $\ell =0$ and let $\theta _{jn}^{\left( \ell \right) }
\overset{iid}{\thicksim }p\left( \theta \right)\quad\left( j\in 
\mathcal{J}\!,\,n\in \mathcal{N}\right) $.

\item For $\ell =1,\dots,L$

\begin{enumerate}
\item Correction $\left( C\right) $ phase:

\begin{enumerate}
\item $w_{jn}\left( t_{\ell-1}\right) =1\quad\left( j\in \mathcal{J}\!,\,n\in \mathcal{N}
\right) $.

\item For $s=t_{\ell -1}+1,\dots,t_{\ell }$%
\begin{equation}
w_{jn}\left( s\right) =w_{jn}\left( s-1\right) \cdot p\left( y_{s}\mid
y_{1:s-1},\theta _{jn}^{\left( \ell -1\right) }\right)\quad\left( j\in 
\mathcal{J}\!,\,n\in \mathcal{N}\right).
\label{C_phase_compute}
\end{equation}

\item $w_{jn}^{\left( \ell -1\right) }:=w_{jn}\left( t_{\ell }\right)\quad
\left( j\in \mathcal{J},n\in \mathcal{N}\right) $.
\end{enumerate}

\item Selection $\left( S\right) $ phase, applied independently to
each group $j\in \mathcal{J}$: Using multinomial or residual sampling based
on $\left\{ w_{jn}^{\left( \ell -1\right) }\,\left( n\in \mathcal{N}
\right) \right\} $, select
\begin{equation*}
\left\{ \theta _{jn}^{\left( \ell ,0\right) }\,\left( n\in \mathcal{N}
\right) \right\} \text{ from }\left\{ \theta _{jn}^{\left( \ell -1\right) }
\,\left( n\in \mathcal{N}\right) \right\}.
\end{equation*}

\item Mutation $\left( M\right) $ phase, applied independently
across $j\in \mathcal{J}\!,\,n\in \mathcal{N}$:
\begin{equation}
\theta _{jn}^{\left( \ell \right) }\thicksim p\left( \theta \mid
y_{1:t_{\ell }},\theta _{jn}^{\left( \ell ,0\right) },\lambda _{\ell
}\right)
\label{Mphase_dist}
\end{equation}
where the drawings are independent and the p.d.f.~\eqref{Mphase_dist}
satisfies the invariance condition
\begin{equation}
\int_{\Theta }
p\left(\theta \mid y_{1:t_{\ell }},\theta ^{\ast },\lambda _{\ell }\right) 
p\left( \theta ^{\ast }\mid y_{1:t_{\ell }}\right) 
d\nu(\theta^{\ast})
= p\left( \theta \mid y_{1:t_{\ell }}\right).
\label{M_invariant}
\end{equation}
\end{enumerate}

\item $\theta _{jn} :=\theta _{jn}^{\left( L\right) }\quad\left( j\in 
\mathcal{J}\!,\,n\in \mathcal{N}\right)$
\end{enumerate}

The algorithm is nonadaptive because $t_{0},\ldots,t_{L}$ and
$\lambda_{1},\ldots,\lambda_{L}$ are predetermined before the algorithm starts.
Going forward it will be convenient to denote the cycle indices by
$\mathcal{L}=\left\{ 1,\ldots ,L\right\} $. At the conclusion of the algorithm,
the simulation approximation of a generic posterior moment is
\eqref{garnd_approx}.

\begin{proposition}
\label{prop_Chopin}
If Conditions \ref{cond:prior_evaluate} through \ref{cond:prior_var} 
are satisfied then Algorithm \ref{alg:general_nonadaptive}
satisfies Condition \ref{cond:normal_approx}.
\end{proposition}

\begin{proof}
The results follow from Chopin (2004), Theorem 1 (for multinomial
resampling) and Theorem 2 (for residual resampling). The assumptions made in
Theorems 1 and 2 are 

\begin{enumerate}
\item $L=T$ and $t_\ell=\ell$ $(\ell=0,\dots,T)$,

\item The functions $p\left( \theta \right) p\left( y_{1:t}\mid \theta
\right)\,(t=1,\ldots ,T) $ are integrable on $\Theta $, and

\item The moments $\mathrm{E}\left( g\left( \theta \right) ^{2+\delta }\mid
y_{1:t}\right)\,(t=1,\ldots ,T) $ exist.
\end{enumerate}

Assumption 1 is merely a change in notation; Conditions
\ref{cond:prior_evaluate} through \ref{cond:bounded_like} imply assumption 2;
and Conditions \ref{cond:prior_evaluate} through \ref{cond:prior_var} imply
assumption 3. Theorems 1 and 2 in Chopin (2004) are stated using the weighted
sample at the end of the last $C$ phase, but as that paper states, they also
apply to the unweighted sample at the end of the following $M$ phase. At the
conclusion of each cycle $\ell$, the $J$ groups of $N$ particles each, $\theta
_{jn}^{\left( \ell \right) }$ $\left( n\in \mathcal{N}\right) $, are mutually
independent because the $S$ phase is executed independently in each group.
\end{proof}

At the end of each cycle $\ell$, all particles $\theta _{jn}^{\left( \ell
\right) }$ are identically distributed with common density $p\left( \theta \mid
y_{1:t_{\ell }}\right)$. Particles in different groups are independent, but
particles within the same group are not. The amount of dependence within groups
depends upon how well the $M$ phases succeed in rejuvenating particle diversity.

\subsection{Adaptive simulators} 
\label{subsec:adaptive}

In fact the sequences $\left\{ t_{\ell}\right\}$ and
$\left\{\lambda_{\ell}\right\}$ for which the algorithm is sufficiently
efficient for actual application are specific to each problem. As a practical
matter these sequences must be tailored to the problem based on the
characteristics of the particles
$\left\{\theta_{jn}^{\left(\ell\right)}\right\}$ produced by the algorithm
itself.  This leads to algorithms of the following type.

\begin{algorithm}
\label{alg:general_adaptive}(Adaptive) Algorithm 
\ref{alg:general_adaptive} is the following generalization of Algorithm 
\ref{alg:general_nonadaptive}. In each cycle $\ell$,
\end{algorithm}

\begin{enumerate}
\item \label{alg:general_adaptive_cond1}$t_{\ell }$ may be random and depend
on $\theta _{jn}^{\left( \ell -1\right) }\,\left( j\in \mathcal{J}\!,\,n\in 
\mathcal{N}\right) $;

\item \label{alg:general_adaptive_cond2}The value of $\lambda _{\ell }$ may
be random and depend on $\theta _{jn}^{\left( \ell ,0\right) }\,\left( j\in 
\mathcal{J}\!,\,n\in \mathcal{N}\right)$.
\end{enumerate}

While such algorithms can be effective and are what is needed for practical
applied work, the theoretical foundation relevant to Algorithm
\ref{alg:general_nonadaptive} does not apply.

The groundwork for Algorithm \ref{alg:general_nonadaptive} was laid by Chopin
(2004).  With respect to the first condition of Algorithm
\ref{alg:general_adaptive}, in Chopin's development each cycle consists of a
single observation, while in practice cycle lengths are often adaptive based on
the effective sample size criterion (Liu and Chen, 1995).  Progress has been
made toward demonstrating Condition \ref{cond:normal_approx} in this case only
recently (Douc and Moulines, 2008; Del Moral et al., 2011), and it appears that
Conditions \ref{cond:prior_evaluate} through \ref{cond:prior_var} are
sufficient for the assumptions made in Douc and Moulines (2008).  With respect
to the second condition of Algorithm \ref{alg:general_adaptive}, demonstration
of Condition \ref{cond:normal_approx} for the adaptations that are necessary to
render sequential posterior simulators even minimally efficient appears to be
well beyond current capabilities.

There are many specific adaptive algorithms, like the one described in Section 
\ref{ss:details} below, that will prove attractive to practitioners
even in the absence of sound theory for appraising the accuracy of posterior
moment approximations. In some cases Condition \ref{cond:normal_approx} will
hold; in others, it will not but the effects will be harmless; and in still
others Condition \ref{cond:normal_approx} will not hold and the relation of 
\eqref{garnd_approx} to \eqref{post_moment} will be an open question.
Investigators with compelling scientific questions are unlikely to forego
attractive posterior simulators awaiting the completion of their proper
mathematical foundations.

The following algorithm resolves this dilemma, providing a basis for implementations 
that are both computationally effective and supported by established theory.

\begin{algorithm}
\label{alg:general_hybrid}(Hybrid)
\end{algorithm}

\begin{enumerate}
\item \label{alg:general_hybrid_step1}Execute Algorithm 
\ref{alg:general_adaptive}.
Retain $\left\{t_{\ell},\lambda_{\ell}\right\}$ $\left(\ell\in\mathcal{L}\right)$ and discard 
$\theta_{jn}$ $\left(j\in\mathcal{J}\!,\,n\in\mathcal{N}\right)$.

\item \label{alg:general_hybrid_step3}Execute Algorithm 
\ref{alg:general_nonadaptive} using the retained 
$\left\{t_{\ell},\lambda_{\ell}\right\}$ but with a new seed for the random number generator used in the
$S$ and $M$ phases.
\end{enumerate}

\begin{proposition}
\label{prop:hybrid}
If Conditions \ref{cond:prior_evaluate} through \ref{cond:prior_var} are
satisfied then Algorithm \ref{alg:general_hybrid} satisfies Condition 
\ref{cond:normal_approx}. 
\end{proposition}

\begin{proof}
Because the sequences $\left\{ t_{\ell }\right\} $ and $\left\{ \lambda
_{\ell }\right\} $ produced in step \ref{alg:general_hybrid_step1} 
are fixed (predetermined) with respect to the
random particles $\theta _{jn}\,\left( j\in \mathcal{J}\!,\,n\in 
\mathcal{N}\right) $ generated in step \ref{alg:general_hybrid_step3}, 
the conditions of
Algorithm \ref{alg:general_nonadaptive} are satisfied in step 
\ref{alg:general_hybrid_step3}. From Proposition \ref{prop_Chopin}, 
Condition \ref{cond:normal_approx} is therefore satisfied.
\end{proof}

\subsection{A specific adaptive simulator} 
\label{ss:details} 

Applications of Bayesian inference are most often directed by investigators who
are not specialists in simulation methodology. A generic version of Algorithm
\ref{alg:general_adaptive} that works well for wide classes of existing and
prospective models and data sets with minimal intervention by investigators is
therefore attractive.  While general descriptions of the algorithm such as
provided in Sections \ref{subsec:nonadaptive}--\ref{subsec:adaptive} give a
useful basis for discussion, what is really needed for practical use is a
detailed specification of the actual steps that need to be executed.  This
section provides the complete specification of a particular variant of the
algorithm that has worked well across a spectrum of models and data sets
characteristic of current research in finance and economics.   The software
package that we are making available with this paper provides a full
implementation of the algorithm presented here (see Section \ref{ss:software}).

\begin{description}
    \item [C phase termination] The $C$ phase is terminated based on a rule
        assessing the effective sample size (Kong et al., 1994; Liu and Chen,
        1995).  Following each update step $s$ as described in Algorithm 1,
        part 2(a)ii, compute 
        $$ \ESS\left( s\right) =\frac{\left[
                \sum_{j=1}^{J}\sum_{n=1}^{N}w_{jn} \left( s\right) \right]
            ^{2}}{\sum_{j=1}^{J}\sum_{n=1}^{N}w_{jn}\left( s\right) ^{2}}.  
        $$
        It is convenient to work with the relative sample size,
        $\RSS=\ESS/(JN)$, where $JN$ is the total number of particles.
        Terminate the $C$ phase and proceed to the $S$ phase if $\RSS<0.5$.

        We have found that the specific threshold value used has little effect
        on the performance of the algorithm.  Higher thresholds imply that the
        $C$ phase is terminated earlier, but fewer Metropolis steps are needed
        to obtain adequate diversification in the $M$ phase (and inversely for
        lower thresholds).  The default threshold in the software is 0.5.  This
        can be easily modified by users, but we have found little reason to do
        so.

    \item [Resampling method] We make several resampling methods available.  The
        results of Chopin (2004) apply to multinomial and residual resampling
        (Baker, 1985, 1987; Liu and Chen, 1998).  Stratified (Kitagawa, 1996)
        and systematic (Carpenter et al., 1999) samplers are also of potential
        interest.  We use residual sampling as the default.  It is
        substantially more efficient than multinomial at little additional
        computational cost.  Stratified and systematic resamplers are yet more
        efficient as well as being less costly and simpler to implement;
        however, we are not aware of any available theory supporting their use
        and so recommend them only for experimental use at this point.

    \item [Metropolis steps] The default is to use a sequence of Gaussian random
        walk samplers operating on the entire parameter vector in a single
        block.  At iteration $r$ of the $M$ phase in cycle $\ell$, the variance
        of the sampler is obtained by $\Sigma_{\ell r} = h^2_{\ell r} V_{\ell
        r}$ where $V_{\ell r}$ is the sample variance of the current
        population of particles and $h_{\ell r}$ is a ``stepsize" scaling
        factor.  The stepsize is initialized at 0.5 and incremented depending
        upon the Metroplis acceptance rate at each iteration.  Using a target
        acceptance rate of 0.25 (Gelman et al., 1996), the stepsize is
        incremented by 0.1 if the acceptance rate is greater than the target
        and decremented by 0.1 otherwise, respecting the constraint $0.1 \leq
        h_{\ell r} \leq 1.0$.  

        An alternative that has appeared often in the literature is to use an
        independent Gaussian sampler. This can give good performance for
        relatively simple models, and is thus useful for illustrative purposes
        using implementations that are less computationally efficient.  But for
        more demanding applications, such as those that are likely to be
        encountered in practical work, the independent sampler lacks
        robustness.  If the posterior is not well-represented by the Gaussian
        sampler, the process of particle diversification can be very slow.  The
        difficulties are amplified in the case of high-dimensional parameter
        vectors.  And the sampler fails spectacularly in the presence of
        irregular posteriors (e.g., multimodality).

        The random walk sampler gives up some performance in the case of near
        Gaussian posteriors (though it still performs well in this situation),
        but is robust to irregular priors, a feature that we consider to be
        important.  The independent sampler is available in the software for
        users wishing to experiment.  Other alternatives for the Metropolis
        updates are also possible, and we intend to investigate some of these
        in future work.                                                           

    \item [M phase Termination]  The simplest rule would be to terminate the $M$
        phase after some fixed number of Metropolis iterations, $R$.  But in
        practice, it sometimes occurs that the $C$ phase terminates with a very
        low $\RSS$.  In this case, the $M$ phase begins with relatively few
        distinct particles and we would like to use more iterations in
        order to better diversify the population of particles.  A better rule
        is to end the $M$ phase after $R_\ell = \kappa\cdot \Rbar$
        iterations if $\RSS<D_2$ and after $R_\ell = \Rbar$ iterations if
        $D_2\leq\RSS<D_1$ (the default settings in the software are $\Rbar=7$,
        $\kappa=3$, $D_1=0.50$ and $D_2=0.20$).  We refer to this as the
        ``deterministic stopping rule."

        But, the optimal settings, especially for $\Rbar$, are highly
        application dependent.  Setting $\Rbar$ too low results in poor
        rejuvenation of particle diversity and the \RNE\ associated with
        moments of functions of interest will typically be low.  This problem
        can be addressed by either using more particles or a higher setting for
        $\Rbar$.  Typically, the latter course of action is more efficient if
        \RNE\ drops below about 0.25; that is, finesse (better rejuvenation) is
        generally preferred over brute force (simply adding more particles).
        On the other hand, there is little point in continuing $M$ phase
        iterations once \RNE\ gets close to one (indicating that particle
        diversification is already good, with little dependence amongst the
        particles in each group).  Continuing past this point is simply a waste
        of computing resources. 
 
        \RNE\ thus provides a useful diagnostic of particle diversity. But it
        also provides the basis for a useful alternative termination criterion.
        The idea is to terminate $M$ phase iterations when the average \RNE\ of
        numerical approximations of selected test moments exceeds some
        threshold.  We use a default threshold of $E_1=0.35$.  It is often
        useful to compute moments of functions of interest at particular
        observation dates.  While it is possible to do this directly using
        weighted samples obtained within the $C$ phase, it is more efficient to
        terminate the $C$ phase and execute $S$ and $M$ phases at these dates
        prior to evaluating moments.  In such cases, it is useful to target a
        higher \RNE\ threshold to assure more accurate approximations.
        We use a default of $E_2=0.90$ here.  We also set an upper bound for
        the number of Metropolis iterations, at which point the $M$ phase
        terminates regardless of \RNE.  The default is $\Rmax=100$.

        While the deterministic rule is simpler, the \RNE-based rule has
        important advantages: it is likely to be more efficient than the
        deterministic rule in practice; and it is fully automated, thus
        eliminating the need for user intervention.   Some experimental results
        illustrating these issues are provided in Section \ref{sec:application}.

        An important point to note here is that after the $S$ phase has
        executed, \RSS\ no longer provides a useful as a measure of the
        diversity of the particle population (sincle the particles are
        equally-weighted).  \RNE\ is the only useful measure of which we are
        aware.

    \item [Parameterization] Bounded support of some parameters can impede the
        efficiency of the Gaussian random walk in the $M$ phase, which has
        unbounded support.  We have found that simple transformations, for
        example a positive parameter $\sigma$ to $\theta = \log(\sigma)$ or a
        parameter $\rho$ defined on $(-1, 1)$ to $\theta = \atanh(\rho)$, can
        greatly improve performance.  This requires attention to the Jacobian
        of the transformation in the prior distribution, but it is often just as
        appealing simply to place the prior distribution directly on the
        transformed parameters.  Section 5.1 provides an example.


%
\end{description}

The posterior simulator described in this section is adaptive (and thus a
special case of Algorithm \ref{alg:general_adaptive}).  That is, key elements
of the simulator are not known in advance, but are determined on the basis of
particles generated in the course of executing it.  In order to use the
simulator in the context of the hybrid method (Algorithm
\ref{alg:general_hybrid}), these design elements must be saved for use in a
second run.  For the simulator described here, we need the observations
$\{t_\ell\}$ at which each $C$ phase termination occurs, the variance matrices
$\{\Sigma_{\ell r}\}$ used in the Metropolis updates, and the number of
Metropolis iterations executed in each $M$ phase $\{R_\ell\}$.   The hybrid
method then involves running the simulator a second time, generating new
particles but using the design elements saved from the first run.  Since the
design elements are now fixed and known in advance, this is a special case of
Algorithm \ref{alg:general_nonadaptive} and thus step 2 of Algorithm
\ref{alg:general_hybrid}.  And therefore, in particular, Condition
\ref{cond:normal_approx} is satisfied and the results of Section
\ref{sec:parpostsims} apply.

\subsection{Software}
\label{ss:software}

The software package that we are making available with this paper
(\url{http://www.quantosanalytics.org/garland/mp-sps_1.1.zip}) implements the
algorithm developed in Section \ref{ss:details}.  The various settings
described there are provided as defaults but can be easily changed by users.
Switches are available allowing the user to choose whether the algorithm runs
entirely on the CPU or using GPU acceleration, and whether to run the algorithm
in adaptive or hybrid mode.


The software provides log marginal likelihood (and log score if a burn-in
period is specified) together with estimates of NSE, as described in Section
\ref{subsec:ML}.  Posterior mean and standard deviation of functions of
interest at specified dates are provided along with corresponding estimates of
RNE and NSE, as described in Section \ref{sec:parpostsims}.  RSS at each $C$
phase update and RNE of specified test moments at each $M$ phase iteration are
available for diagnostic purposes, and plots of these are generated as part of
the standard output.

For transparency and accessibility to economists doing practical applied work,
the software uses a Matlab shell (with extensive use of the parallel/GPU
toolbox), with calls to functions coded in C/CUDA for the computationally
intensive parts.  For simple problems where computational cost is low to begin
with, the overhead associated with using Matlab is noticeable.  But, for
problems where computational cost is actually important, the cost is nearly
entirely concentrated in C/CUDA code.  Very little efficiency is lost by using
a Matlab shell, while the gain in transparency is substantial.

The software is fully object-oriented, highly modular, and easily extensible.
It is straightforward for users to write plug-ins implementing new models,
providing access to the full benefits of GPU acceleration with little
programming effort.  All that is needed is a Matlab class with methods that
evaluate the data density and moments of any functions of interest.  The
requisite interfaces are all clearly specified by abstract classes.  

To use a model for a particular application, data and a description of the
prior must be supplied as inputs to the class constructor. Classes implementing
simulation from and evaluation of some common priors are included with the
software.  Extending the package to implement additional prior specifications
is simply a matter of providing classes with this functionality.

The software includes several models (including the application described in
Section \ref{sec:application} of this paper), and we are in the process of
extending the library of available models and illustrative applications using
this framework in ongoing work.   Geweke et al.~(2013) provides an example of
one such implementation for the logit model.

\section{Predictive and marginal likelihood}
\label{sec:PML}

Marginal likelihoods are fundamental to the Bayesian comparison, averaging and
testing of models. Sequential posterior simulators provide
simulation-consistent approximations of marginal likelihoods as a
by-product---almost no additional computation is required.  This happens
because the marginal likelihood is the product of predictive likelihoods over
the sample, and the $C$ phase approximates these terms in the same way that
classical importance sampling can be used to approximate marginal likelihood
(Kloek and van Dijk, 1978; Geweke, 2005, Section 8.2.2). By contrast there is
no such generic approach with MCMC posterior simulators.

The results for posterior moment approximation apply only indirectly to
marginal likelihood approximation. This section describes a practical
adaptation of the results to this case, and a theory of approximation that
falls somewhat short of Condition \ref{cond:normal_approx} for posterior
moments.

\subsection{Predictive likelihood}
\label{subsec:PL}

Any predictive likelihood can be cast as a posterior moment by defining
\begin{equation}
g\left( \theta ;t,s\right) =p\left( y_{t+1:s}\mid y_{1:t},\theta \right)\quad (s>t).
\label{PL_func}
\end{equation}
Then
\begin{equation}
p\left( y_{t+1:s}\mid y_{1:t}\right) 
   =\mathrm{E}\left[g\left(\theta;t,s\right)\mid y_{1:t}\right] 
   :=\overline{g}\left(t,s\right).
\label{PL__formula}
\end{equation}

\begin{proposition}
\label{prop_PL}If Conditions \ref{cond:prior_evaluate}, 
\ref{cond:LF_evaluate} and \ref{cond:bounded_like} are satisfied, then in Algorithms
\ref{alg:general_nonadaptive} and \ref{alg:general_hybrid} 
Condition \ref{cond:normal_approx} is satisfied
for the function of interest (\ref{PL_func}) and the sample $y_{1:t}$.
\end{proposition}

\begin{proof}
It suffices to note that Conditions \ref{cond:prior_evaluate} through 
\ref{cond:bounded_like} imply that Condition \ref{cond:prior_var} is satisfied
for \eqref{PL_func}.
\end{proof}

The utility of Proposition \ref{prop_PL} stems from the fact that simulation
approximations of many moments (\ref{PL__formula}) are computed as
by-products in Algorithm \ref{alg:general_nonadaptive}. Specifically, in the
$C$ phase at \eqref{C_phase_compute} for  $t=t_{\ell-1}$ and
$s=t_{\ell-1}+1,\dots,t_{\ell}$ 
\begin{equation}
w_{jn}\left( s\right) 
   =g\left(\theta_{jn}^{\left(\ell-1\right)};t,s\right)
   \quad\left(j\in\mathcal{J}\!,\,n\in\mathcal{N}\!,\,\ell\in\mathcal{L}\right).
\label{wt_PL}
\end{equation}

From \eqref{CLT} Proposition \ref{prop_PL} applied over the sample $y_{1:t_{\ell-1}}$
implies for $t=t_{\ell-1}$ and $s>t_{\ell-1}$
\begin{equation*}
\overline{g}^{(J,N)}(t,s) 
 = \left( JN\right) ^{-1}\sum_{j=1}^{J}\sum_{n=1}^{N}
   g\left(\theta_{jn}^{\left(\ell-1\right)};t,s\right)
  \overset{p}{\rightarrow} p\left( y_{t+1:s}\mid y_{1:t}\right)\quad(\ell\in\mathcal{L}). 
\end{equation*}
Algorithms \ref{alg:general_nonadaptive} and \ref{alg:general_hybrid}
therefore provide weakly consistent approximations to
\eqref{PL__formula} expressed in terms of the form \eqref{wt_PL}. The logarithm
of this approximation is a weakly consistent approximation of the more commonly
reported log predictive likelihood.  
The numerical standard error of $\overline{g}^{(J,N)}(t,s)$ is given by
\eqref{NSE_def}, and the delta method provides the numerical standard error for
the log predictive likelihood.


The values of $s$ for which this result is useful are limited because as $s$
increases, $g\left( \theta;t,s\right) $ becomes increasingly concentrated and
the computational efficiency of the approximations $\overline{g}^{\left(
J,N\right) }\left( t,s\right) $ of $\overline{g}\left( t,s\right) $ declines.
Indeed, this is why particle renewal (in the form of $S$ and $M$ phases) is
undertaken when the effective sample size drops below a specified threshold.


The sequential posterior simulator provides access to the posterior
distribution at each observation in the $C$ phase using the particles with the
weights computed at \eqref{C_phase_compute}.  If at this point one executes the
auxiliary simulations $Y_{sjn}^{\left( \ell
-1\right) }\thicksim p\left( Y_{s}\mid
y_{1:s-1},\theta_{jn}^{\left(\ell-1\right)}\right)$ $\left(
j\in\mathcal{J}\!,\,n\in\mathcal{N}\right)$ then a simulation-consistent
approximation of the cumulative distribution of a function $F\left(
Y_{s}\right) $, evaluated at the observed value $F\left( y_{s}\right) $, is
\begin{equation*}
\frac{\sum_{j=1}^{J}\sum_{n=1}^{N}w_{jn}\left( s-1\right) I_{\left( -\infty
,F\left( y_{s}\right) \right] }\left[ F\left( Y_{sjn}^{\left( \ell -1\right)
}\right) \right] }{\sum_{j=1}^{J}\sum_{n=1}^{N}w_{jn}\left( s-1\right)}.
\end{equation*}
These evaluations are the essential element of a probability integral
transform test of model specification (Rosenblatt, 1952; Smith, 1985;
Diebold et al., 1998; Berkowitz, 2001; Geweke and Amisano, 2010).

Thus by accessing the weights $w_{jn}\left( s\right) $ from the $C$ phase of
the algorithm, the investigator can compute simulation consistent
approximations of any set of predictive likelihoods (18) and can execute
probability integral transform tests for any function of $Y_{t}.$

\subsection{Marginal likelihood} 
\label{subsec:ML}

To develop a theory of sequential posterior simulation approximation to the
marginal likelihood, some extension of the notation is useful. Denote
\begin{align*}
\overline{w}_{j}^{N}\left( \ell -1\right) 
     &=N^{-1}\sum_{n=1}^{N}w_{jn}^{\left( \ell -1\right) }
     \quad (\ell \in \mathcal{L},\,j\in \mathcal{J}),\\
\overline{w}_{j}^{N} 
     &=\mathop{\displaystyle\prod }\limits_{\ell=1}^{L}\overline{w}_{j}^{N}\left(\ell-1\right)
     \quad (j\in\mathcal{J}),\\
\overline{w}^{\left( J,N\right) }\left( \ell -1\right)
     &=J^{-1}\sum_{j=1}^{J}\overline{w}_{j}^{N}\left( \ell -1\right)
     \quad(\ell\in\mathcal{L});
\end{align*}
and then 
\begin{equation*}
\overline{w}^{\left( J,N\right) }=J^{-1}\sum_{j=1}^{J}\overline{w}_{j}^{N}
\qquad\widetilde{w}^{\left( J,N\right) }=\mathop{\displaystyle \prod }
\limits_{\ell =1}^{L}\overline{w}^{\left( J,N\right) }\left( \ell -1\right). 
\end{equation*}

It is also useful to introduce the following condition, which
describes an ideal situation that is not likely to be attained in 
practical applications but is useful for expositional purposes.
\begin{condition}
\label{cond:Mphase}In the mutation phase of Algorithm 
\ref{alg:general_nonadaptive}
\begin{equation*}
p\left( \theta \mid y_{1:t_{\ell }},\theta ^{\ast },\lambda \right) =p\left(
\theta \mid y_{1:t_{\ell }}\right) \text{ }\forall \text{ }\theta ^{\ast
}\in \Theta \text{.}
\end{equation*}
\end{condition}
With the addition of Condition \ref{cond:Mphase} it would follow that 
$v=\mathrm{var}\left[ g\left( \theta \right) \mid y_{1:T}\right] $ in Condition
\ref{cond:normal_approx}, and
computed values of relative numerical efficiencies would be about 1 for all
functions of interest $g\left( \theta \right)$. In general Condition 
\ref{cond:Mphase} is unattainable in any interesting application of an
sequential posterior simulator, for if it were the posterior distribution
could be sampled by direct Monte Carlo.

\begin{proposition}
\label{prop_ML}If Conditions \ref{cond:prior_evaluate}, 
\ref{cond:LF_evaluate} and \ref{cond:bounded_like} are satisfied then in
Algorithms \ref{alg:general_nonadaptive} and \ref{alg:general_hybrid} 
\begin{equation}
\overline{w}^{\left( J,N\right) }\overset{p}{\rightarrow }p\left(
y_{1:T}\right) ,\ \widetilde{w}^{\left( J,N\right) }\overset{p}{\rightarrow }
p\left( y_{1:T}\right),
\label{weakp_ML}
\end{equation}
as $N\rightarrow \infty $, and 
\begin{equation}
\mathrm{E}\left\{ \left[ J\left( J-1\right) \right] ^{-1}\sum_{j=1}^{J}
\left( \overline{w}_{j}^{N}-\overline{w}^{\left( J,N\right) }\right)
^{2}\right\} =\mathrm{var}\left( \overline{w}^{\left( J,N\right) }\right). 
\label{prop_4A}
\end{equation}
If, in addition, Condition \ref{cond:Mphase} is satisfied then
\begin{equation}
\left( JN\right) ^{-1/2}\left[ \overline{w}^{\left( J,N\right) }-p\left(
y_{1:T}\right) \right] \overset{d}{\rightarrow }N\left( 0,v\right).
\label{prop_4B}
\end{equation}
\end{proposition}

\begin{proof}
From Proposition \ref{prop_PL}, Condition \ref{cond:normal_approx} is
satisfied for $g\left( \theta ;t_{\ell -1},t_{\ell }\right)\,
\left( \ell \in \mathcal{L}\right) $. Therefore $\overline{w}_{j}^{N}\left(
\ell -1\right) \overset{p}{\rightarrow }p\left( y_{t_{\ell -1}+1:t_{\ell
}}\right) $ as $N\rightarrow \infty $ $\left( j\in \mathcal{J};\ell \in 
\mathcal{L}\right) $. The result \eqref{weakp_ML} follows from the
decomposition
\begin{equation*}
p\left( y_{1:T}\right) =\mathop{\displaystyle \prod }\limits_{\ell
=1}^{L}p\left( y_{t_{\ell -1}+1:t_{\ell }}\mid y_{1:t_{\ell -1}}\right). 
\end{equation*}

Condition \ref{cond:bounded_like} implies that moments of $\overline{w}_{j}^{N}$ 
of all orders exist. Then \eqref{prop_4A} follows from the mutual
independence of $\overline{w}_{1}^{N},\ldots ,\overline{w}_{J}^{N}$.

Result \eqref{prop_4B} follows from the mutual independence and asymptotic
normality of the $JL$ terms $\overline{w}_{j}^{N}\left(\ell\right)$, 
and $v$ follows from the usual asymptotic expansion of the product in
$\overline{w}_{j}^{N}$.  
\end{proof}

Note that without Condition \ref{cond:Mphase}, $w_j^N(\ell)$and $w_j^N(\ell')$
are not independent and so Condition \ref{cond:normal_approx} does not provide
an applicable central limit theorem for the product $\overline{w}_j^N$
(although Condition \ref{cond:normal_approx} does hold for each $\overline{w}_j^N(\ell-1)$
$(j\in\mathcal{J}\!,\,\ell\in\mathcal{L})$ individually, as shown in Section
\ref{subsec:PL}).

Let $\vhat\left( \overline{w}^{\left( J,N\right) }\right) $ denote the term in
braces in \eqref{prop_4A}.  Proposition \ref{prop_ML} motivates the working
approximations
\begin{eqnarray}
&&\overline{w}^{\left( J,N\right) }\overset{\cdot }{\thicksim }N\left[
p\left( y_{1:T}\right),\vhat\left( \overline{w}^{\left( J,N\right)
}\right) \right],\notag \\
&&\log \left( \overline{w}^{\left( J,N\right) }\right) \overset{\cdot }
{\thicksim }N\left[ \log p\left( y_{1:T}\right) ,
\frac{\vhat\left( \overline{w}^{\left( J,N\right) }\right) }
{\left( \overline{w}^{(J,N)}\right) ^{2}}\right].
\label{ML_se}
\end{eqnarray}
Standard expansions of $\log \overline{w}^{N}$ and $\log \widetilde{w}^{N}$
suggest the same asymptotic distribution, and therefore comparison of these
values in relation to the working standard error from (\ref{ML_se}) provides
one indication of the adequacy of the asymptotic approximation.  Condition
\ref{cond:Mphase} suggests that these normal approximations will be more
reliable as more Metropolis iterations are undertaken in the algorithm detailed
in Section \ref{ss:details}, and more generally, the more effective is the
particle diversity generated by the $M$ phase.

\section{Application: Exponential generalized
autoregressive conditional heteroskedasticity model}
\label{sec:application}

The predictive distributions of returns to financial assets are central to the
pricing of their derivatives like futures contracts and options. The literature
modeling asset return sequences as stochastic processes is enormous and has
been a focus and motivation for Bayesian modelling in general and application
of sequential Monte Carlo (SMC) methods in particular. One of these models is
the exponential generalized autoregressive conditional heteroskedasticity
(EGARCH) model introduced by Nelson (1991). The example in this section works
with a family of extensions developed in Durham and Geweke (2013) that is
highly competitive with many stochastic volatility models.

In the context of this paper the EGARCH model is also of interest because its
likelihood function is relatively intractable. The volatility in the model is
the sum of several factors that are exchangeable in the posterior distribution.
The return innovation is a mixture of normal distributions that are also
exchangeable in the posterior distribution. Both features are essential to the
superior performance of the model (Durham and Geweke, 2013). Permutations in
the ordering of factors and mixture components induce multimodal distributions
in larger samples.  Models with these characteristics have been widely used as
a drilling ground to assess the performance of simulation approaches to
Bayesian inference with ill-conditioned posterior distributions (e.g., Jasra et
al., 2007). The models studied in this section have up to $\left( 4!\right)
^{2}=576$ permutations, and potentially as many local modes.  Although it would
be possible to optimize the algorithm for these characteristics, we
intentionally make no efforts to do so.  Nonetheless, the irregularity of the
posteriors turns out to pose no difficulties for the algorithm.

Most important, in our view, this example illustrates the potential large
savings in development time and intellectual energy afforded by the algorithm
presented in this paper compared with other approaches that might be taken. We
believe that other existing approaches, including importance sampling and
conventional variants on Markov chain Monte Carlo (MCMC), would be substantially more
difficult. At the very least they would require experimentation with tuning
parameters by Bayesian statisticians with particular skills in these numerical
methods, even after using the approach of Geweke (2007) to deal with dimensions
of posterior intractability driven by exchangeable parameters in the posterior
distribution. The algorithm replaces this effort with a systematic updating
of the posterior density, thereby releasing the time of investigators for more
productive and substantive efforts.

\subsection{Model and data}

An EGARCH model for a sequence of asset returns $\left\{ y_{t}\right\} $ has
the form
\begin{align}
v_{kt} &= \alpha _{k}v_{k,t-1}+\beta _{k}\left( \left\vert \varepsilon
_{t-1}\right\vert -\left( 2/\pi \right) ^{1/2}\right) +\gamma
_{k}\varepsilon _{t-1}\quad\left( k=1,\ldots ,K\right),
\label{EGARCH_mod1} \\
y_{t} &= \mu _{Y}+\sigma _{Y}\exp \left( \sum_{k=1}^{K}v_{kt}/2\right)
\varepsilon _{t}.
\label{EGARCH_mod2}
\end{align}
The return disturbance term $\varepsilon_{t}$ is distributed as a mixture
of $I$ normal distributions, 
\begin{equation}
p(\varepsilon _{t})=\sum_{i=1}^{I}p_{i}\phi (\epsilon_{t};\mu _{i},\sigma _{i}^{2})
\label{normal_mixture}
\end{equation}%
where $\phi (\,\cdot\,;\mu,\sigma^{2})$ is the Gaussian density with
mean $\mu$ and variance $\sigma^{2}$, $p_{i}>0$ $( i=1,\ldots,I)$ and 
$\sum_{i=1}^{I}p_{i}=1$. The parameters of the model are
identified by the conditions $\mathrm{E}\left( \varepsilon _{t}\right) =0$
and $\var\left(\varepsilon_{t}\right)=1$; equivalently,
\begin{equation}
\sum_{i=1}^{I}p_{i}\mu _{i}=0,\qquad\sum_{i=1}^{I}p_{i}(\mu_{i}^{2}+\sigma_{i}^{2})=1.
\label{EGARCH_normalization}
\end{equation}
The models are indexed by $K$, the number of volatility factors, and $I$,
the number of components in the return disturbance normal mixture, and we
refer to the specification (\ref{EGARCH_mod1})--(\ref{EGARCH_mod2}) as \egarch{KI}.
The original form of the EGARCH model (Nelson, 1991)
is \eqref{EGARCH_mod1}--\eqref{EGARCH_mod2} with $I=K=1$.

\begin{table}
\caption{Parameters and prior distributions for the EGARCH models}
\smallskip
\small
All parameters have Gaussian priors with means and standard deviations 
indicated below (the prior distribution of $\theta_{8i}$ is
truncated below at $-3.0$).  Indices $i$ and $k$ take on the 
values $i=1,\dots,I$ and $k=1,\dots,K$.
\medskip

\centering
\begin{tabular}{cccc}
\hline\hline
& Mean & Std Dev & Transformation \\ 
\hline
$\theta _{1}$  & $0$ & $1$ & $\mu _{Y}=\theta_{1}/1000$ \\ 
$\theta _{2}$  & $\log (0.01)$ & $1$ & $\sigma _{Y}=\exp (\theta _{2})$ \\ 
$\theta _{3k}$ & $\tanh ^{-1}(0.95)$ & $1$ & $\alpha _{k}=\tanh (\theta_{3k})$ \\ 
$\theta _{4k}$ & $\log (0.10)$ & $1$ & $\beta _{k}=\exp (\theta _{4k})$ \\ 
$\theta _{5k}$ & $0$ & $0.2$ & $\gamma _{k}=\theta _{5k}$ \\ 
$\theta _{6i}$ & $0$ & $1$ & $p_{i}^{\ast }=\tanh (\theta _{6i})+1$ \\ 
$\theta _{7i}$ & $0$ & $1$ & $\mu _{i}^{\ast }=\theta _{7i}$ \\ 
$\theta _{8i}$ & $0$ & $1$ & $\sigma _{i}^{\ast }=\exp (\theta _{8i})$ \\ 
\hline
\end{tabular}
\label{tab:EGARCH_prior}
\end{table}

The algorithm operates on transformed parameters, as described in Section
\ref{ss:details}.  The vector of transformed parameters is denoted $\theta$.
The prior distributions for the elements of $\theta$ are all Gaussian.  Priors
and transformations are detailed in Table \ref{tab:EGARCH_prior}.  The
intermediate parameters $p_{i}^{\ast}$, $\mu_{i}^{\ast} $ and
$\sigma_{i}^{\ast}$ are employed to enforce the normalization
\eqref{EGARCH_normalization} and undergo the further transformations
\begin{gather}
p_{i}=p_{i}^{\ast }/\sum_{i=1}^{I}p_{i}^{\ast }\text{, \ }\mu _{i}^{\ast
\ast }=\mu _{i}^{\ast }-\sum_{i=1}^{I}p_{i}\mu _{i}^{\ast }\text{, \ }c=%
\left[ \sum_{i=1}^{I}p_{i}((\mu _{i}^{\ast \ast })^{2}+(\sigma _{i}^{\ast
})^{2})\right] ^{-1/2}, 
\label{EGARCH_tran1} \\
\mu _{i}=c\mu _{i}^{\ast \ast }\text{, \ }\sigma_{i}=c\sigma _{i}^{\ast }%
\text{ }\left( i=1,\ldots ,I\right).
\label{EGARCH_tran2}
\end{gather}
The truncation of the parameters $\theta_{8i}$ $(i=1,\dots,I)$ bounds the likelihood
function above, thus satisfying Condition \ref{cond:bounded_like}. 
The initial simulation from the prior distribution is trivial, as is
evaluation of the prior density.


Evaluation of $p\left( y_{t}\mid y_{1:t-1},\theta \right) $ entails the
following steps, which can readily be expressed in SIMD-compatible code,
satisfying Condition \ref{cond:LF_evaluate}.
\begin{enumerate}
\item Transform the parameter vector $\theta$ to the parameters of the model
\eqref{EGARCH_mod1}--\eqref{EGARCH_mod2} using the fourth column of Table
\ref{tab:EGARCH_prior} and \eqref{EGARCH_tran1}--\eqref{EGARCH_tran2}.

\item Compute $v_{kt}$ $\left( k=1,\ldots ,K\right) $ using
\eqref{EGARCH_mod1}, noting that $\varepsilon _{t-1}$ and $v_{k,t-1}$ $\left(
k=1,\ldots ,K\right) $ are available from the evaluation of $p\left(
y_{t-1}\mid y_{1:t-2},\theta \right)$. As is conventional in these models the
volatility states are initialized at $v_{k0}=0$ $\left( k=1,\ldots ,K\right) $.

\item Compute $h_{t}=\sigma _{Y}\exp \left( \sum_{k=1}^{K}v_{kt}/2\right) $
and $\varepsilon _{t}=\left( y_{t}-\mu _{Y}\right) /h_{t}$.

\item Evaluate $p\left( y_{t}\mid y_{1:t-1},\theta \right) =
  (2\pi)^{-1/2}h_{t}^{-1} 
  \sum_{i=1}^{I}\left\{p_i\frac{1}{\sigma_i}
     \exp\left[-\left(\varepsilon_t-\mu_i\right)^{2}/2\sigma_{i}^{2}\right]\right\}$.
\end{enumerate}

The observed returns are $y_{t}=\log \left( p_{t}/p_{t-1}\right) $ $\left(
t=1,\ldots ,T\right) $ where $p_{t}$ is the closing Standard and Poors 500
index on trading day $t$. We use returns beginning January 3, 1990 $\left(
t=1\right)$ and ending March 31, 2010 $\left( t=T=5100\right)$.

\subsection{Performance}
\label{subsec:EGARCH_performance}

All of the inference for the EGARCH models is based on $2^{16}=65,536$
particles in $J=2^{6}=64$ groups of $N=2^{10}=1024$ particles each.  Except
where specified otherwise, we use the $C$ phase stopping rule with $D_1=0.50$
and $D_2=0.20$, and the \RNE-based $M$ phase stopping rule with $E_1=0.35$,
$E_2=0.90$ and $\Rmax=300$.  Unless otherwise noted, reported results are
obtained from step 2 of the hybrid algorithm. See Section \ref{ss:details} for
details.

For the \RNE-based rule, we use the following test functions: log volatility,
$g_1(\theta,y_{1:t})$ $=$ $\log\sigma_Y + \sum_{k=1}^Kv_{kt}/2$; skewness of
the mixture distribution \eqref{normal_mixture}, $g_2(\theta,y_{1:t})$ $=$
\allowbreak{} $\E(\epsilon^3_{t+1}|\theta)$; and 3\% loss probability,
$g_3(\theta,y_{1:t}) = P(Y_{t+1}<-0.03|\theta,y_{1:t})$.  Posterior means of
these test functions are of potential interest in practice, and we return to a
more detailed discussion of these applications in Section \ref{ss:moments}.

\begin{table}         
\caption{Comparison of EGARCH models.}
\label{tab:EGARCHcomp}
\vspace{1ex}
\small
\centering
    \begin{tabular}{ccccccccccccc}
\hline\hline        
Model & Hybrid  & Compute  & Cycles & Metropolis & Log &  NSE  & Log   &  NSE  \\
      & Step    & Time     &        & Steps      & ML  &       & Score &       \\
      &         & (Seconds) \\
\hline
\egarch{11}  &   1  &      74  &     54  &    515  &  16,641.92  & 0.1242  & 15,009.20  & 0.1029\\        
\egarch{11}  &   2  &      65  &     54  &    515  &  16,641.69  & 0.0541  & 15,009.08  & 0.0534\\[.5ex]  
\egarch{12}  &   1  &     416  &     65  &   2115  &  16,713.44  & 0.0814  & 15,075.82  & 0.0596\\        
\egarch{12}  &   2  &     812  &     65  &   2115  &  16,713.60  & 0.0799  & 15,075.91  & 0.0649\\[.5ex]  
\egarch{21}  &   1  &     815  &     62  &   4439  &  16,669.40  & 0.0705  & 15,038.40  & 0.0630\\        
\egarch{21}  &   2  &     732  &     62  &   4439  &  16,669.39  & 0.0929  & 15,038.41  & 0.0887\\[.5ex]  
\egarch{22}  &   1  &    1104  &     71  &   4965  &  16,736.81  & 0.0704  & 15,100.17  & 0.0534\\        
\egarch{22}  &   2  &     991  &     71  &   4965  &  16,736.89  & 0.0864  & 15,100.25  & 0.0676\\[.5ex]  
\egarch{23}  &   1  &    2233  &     77  &   7490  &  16,750.77  & 0.0683  & 15,114.24  & 0.0455\\        
\egarch{23}  &   2  &    2093  &     77  &   7490  &  16,750.83  & 0.0869  & 15,114.21  & 0.0512\\[.5ex]  
\egarch{32}  &   1  &    1391  &     76  &   6177  &  16,735.04  & 0.0870  & 15,099.91  & 0.0650\\        
\egarch{32}  &   2  &    1276  &     76  &   6177  &  16,734.94  & 0.0735  & 15,099.90  & 0.0540\\[.5ex]  
\egarch{33}  &   1  &    2685  &     82  &   8942  &  16,748.74  & 0.0703  & 15,113.62  & 0.0397\\        
\egarch{33}  &   2  &    2619  &     82  &   8942  &  16,748.75  & 0.0646  & 15,113.68  & 0.0456\\[.5ex]  
\egarch{34}  &   1  &    3036  &     82  &   8311  &  16,748.78  & 0.0671  & 15,113.76  & 0.0486\\        
\egarch{34}  &   2  &    2878  &     82  &   8311  &  16,748.64  & 0.0716  & 15,113.64  & 0.0413\\[.5ex]  
\egarch{43}  &   1  &    2924  &     79  &   9691  &  16,745.62  & 0.0732  & 15,112.36  & 0.0462\\        
\egarch{43}  &   2  &    2741  &     79  &   9691  &  16,745.61  & 0.0725  & 15,112.41  & 0.0534\\[.5ex]  
\egarch{44}  &   1  &    3309  &     82  &   9092  &  16,745.63  & 0.1025  & 15,112.33  & 0.0451\\        
\egarch{44}  &   2  &    3133  &     82  &   9092  &  16,745.54  & 0.0643  & 15,112.31  & 0.0508\\        
\hline\hline
\end{tabular}
\end{table}

The exercise begins by comparing the log marginal likelihood of the 10 variants
of this model indicated in column 1 of Table \ref{tab:EGARCHcomp}.  The log
marginal likelihood and its numerical standard error (\NSE) are computed as
described in Section \ref{subsec:ML}.  \textquotedblleft Log
score\textquotedblright\ is the log predictive likelihood $p\left(
y_{505:5100}\mid y_{1:504}\right) $; observation 505 is the return on the first
trading day of 1992.

The table provides results in pairs corresponding to the two steps of the
hybrid algorithm (Algorithm \ref{alg:general_hybrid}).  In step 1, numerical
approximations are based on particles generated using the simulator detailed in
Section \ref{ss:details} run adaptively.  In step 2,  new particles are
obtained by rerunning the algorithm of Section \ref{ss:details} (nonadaptively)
using design elements retained from the first run.  

There are two striking features in each pair of results: (1) numerical standard
errors are mutually consistent given the reliability inherent with $J=64$
groups of particles; (2) differences in log score or marginal likelihood are,
in turn, consistent with these numerical standard errors.  These findings are
consistent with two conjectures: (1) the central limit theorem for posterior
moments (Proposition \ref{prop:CLT_var_approx}) applies to the simulator
detailed in Section \ref{ss:details} when run adaptively as well as when run
nonadaptively in the context of the hybrid method; (2) conclusion
\eqref{prop_4B} of Proposition \ref{prop_ML} does not require Condition
\ref{cond:Mphase}, which is ideal rather than practical, and is true under
Conditions \ref{cond:prior_evaluate} through \ref{cond:bounded_like}, which are
weak and widely applicable.

We emphasize that these are working conjectures. They are not theorems that are
likely to be proved any time in the near future, if ever.  Our current
recommendation for practical application, pending more extensive experience
with these matters, is to run the simulator adaptively in the work of model
development and modification; then check results using the full hybrid method
(step 2 of Algorithm 3) at regular intervals to guard against unpleasant
surprises; and to always do so before making results available publicly.

Bayes factors strongly favor the \egarch{23} model, as do log scores. In
comparison with all the models that it nests, the Bayes factor is at least
$\exp\left(15\right) $. In comparison with all the models that nest
\egarch{23}, the Bayes factor ranges from $\exp\left( 2\right) $ to over
$\exp\left(5\right)$. This pattern is classic: the data provide strong
diagnostics of underfitting, while the evidence of overfitting is weaker
because it is driven primarily by the prior distribution. The 95\% confidence
intervals for log Bayes factors are generally shorter than 0.2.  Going forward,
all examples in this section utilize the \egarch{23} model.

Compute time increases substantially as a function of $I$ and $K$ in the model
specification \eqref{EGARCH_mod1}--\eqref{normal_mixture}. There are three
reasons for this: the larger models require more floating point operations to
evaluate conditional data densities; the larger models exhibit faster reduction
of effective sample size in the $C$ phase, increasing the number of cycles $L$
and thereby the number of passes through the $M$ phase; and the larger models
require more Metropolis iterations in each $M$ phase to meet the \RNE\
threshold for the stopping rule.

\begin{table}    
\caption{Sensitivity to number of Metropolis steps in $M$ phase.}
\vspace{1ex}
\small
\centering
\begin{tabular}{ccccccccccc}
\hline\hline    
$\Rbar$ &  Compute & Cycles & Metropolis & Log            & NSE      &  Precision  \\
        &    Time  &        & Steps      & Score          &          &   / Time    \\
        &  (Seconds) \\
\hline
   5    &      91  &     77  &    395  &  15,113.13  &   0.4597   &     0.052\\
   8    &     127  &     75  &    616  &  15,114.38  &   0.5441   &     0.027\\
  13    &     205  &     78  &   1040  &  15,114.02  &   0.2888   &     0.059\\
  21    &     304  &     76  &   1638  &  15,114.39  &   0.1880   &     0.093\\
  34    &     482  &     76  &   2652  &  15,114.27  &   0.1277   &     0.127\\
  55    &     776  &     77  &   4345  &  15,114.23  &   0.0849   &     0.179\\
  89    &    1245  &     77  &   7031  &  15,114.16  &   0.0592   &     0.229\\
 144    &    2120  &     79  &  11664  &  15,114.27  &   0.0395   &     0.302\\
   *    &    2329  &     77  &   8531  &  15,114.24  &   0.0427   &     0.236\\
\hline\hline
\end{tabular}
\label{t:rbar}
\end{table}

Table \ref{t:rbar} shows the outcome of an exercise exploring the relationship
between the number of Metropolis iterations undertaken in the $M$ phases and
the performance of the algorithm.  The first 8 rows of the table use the
deterministic $M$ phase stopping rule with $\Rbar$ = 5, 8, 13, 21, 34, 55, 89
and 144.  The last row of the table, indicated by `*' in the $\Rbar$ field,
uses the \RNE-based stopping rule with $E_1=0.35$, $E_2=0.90$ and $\Rmax=300$.
The last column of the table reports a measure of computational performance:
precision (i.e., $1/{\NSE\,}^2$) normalized by computational cost (time in
seconds).


The ratio of precision to time is the relevant criterion for comparing two
means of increasing numerical accuracy: more steps in the $M$ phase versus more
particles.  The ratio of precision to time is constant in the latter strategy.
(In fact, on a GPU, it increases up to a point because the GPU works more
efficiently with more threads, but the application here with $2^{16}$ particles
achieves full efficiency.)  Therefore adding more steps in the $M$ phase is the
more efficient strategy so long as the ratio of precision to time continues to
increase.  Table \ref{t:rbar} shows that in this example, adding iterations is
dominant at $\Rbar=89$ and is likely dominant at $\Rbar=144$.

Ideally, one would select $\Rbar$ to maximize the performance measure reported
in the last column of the table.  The \RNE-based stopping rule, shown in the
last line of the table, does a good job of automatically picking an appropriate
stopping point without requiring any user input or experimentation.  The number
of Metropolis iterations varies from one $M$ phase to the next with this
stopping rule, averaging just over 100 in this application.  Although the
RNE-based rule uses fewer Metropolis steps overall relative to $\Rbar=144$,
total execution time is greater.  This is because the Metropolis iterations are
more concentrated toward the end of the sample period with the \RNE-based rule,
where they are more computationally costly.

The fact that so many Metroplis iterations are needed to get good particle
diversity is symptomatic of the extremely irregular posterior densities implied
by this model (we return to this issue in Section \ref{ss:irregular}).  For
models with posterior densities that are closer to Gaussian, many fewer
Metropolis iterations will typically be needed.



\subsection{Posterior moments}
\label{ss:moments}

Models for asset returns, like the EGARCH models considered here, are primarily
of interest for their predictive distributions.  We illustrate this application
using the three functions of interest $g_i(\theta,y_{1:t})$ introduced in Section
\ref{subsec:EGARCH_performance}.

Moments are evaluated by Monte Carlo approximation over the posterior
distribution of $\theta$ using the particles obtained at time $t$, as described
in Section \ref{sec:parpostsims}.  The last observation of the sample,
March 31, 2010 in this application ($t=5100$), is typically of interest.  For
illustration, we specify the same date one year earlier, March 31, 2009 ($t=4848$), as an
additional observation of interest.  Volatility is much higher on the
earlier date than it is on the later date.    At each date of interest, the $C$
phase is terminated (regardless of the \RSS) and $S$ and $M$ phases are
executed.

\begin{figure} 
    \includegraphics[width=\textwidth]{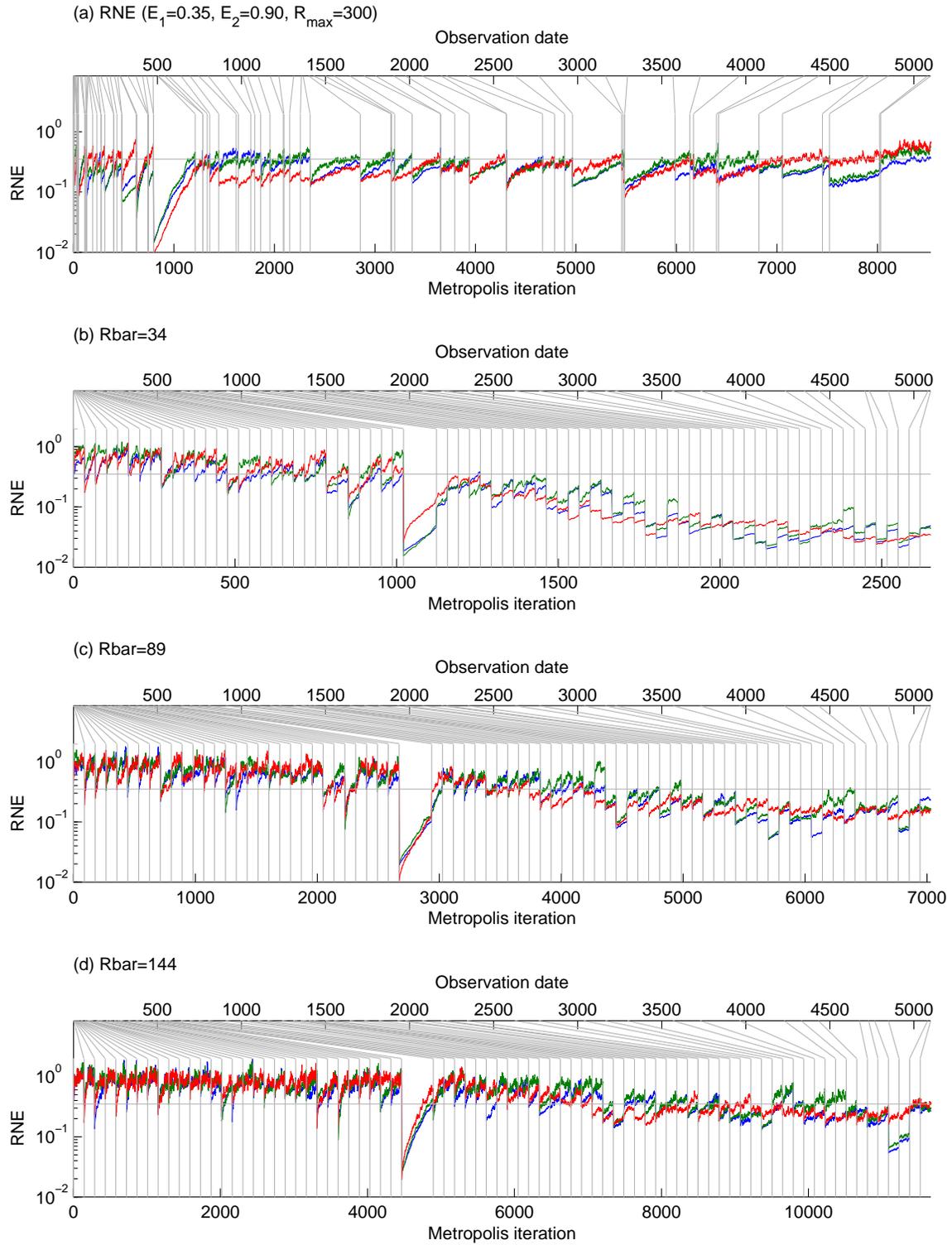} 
    \caption{\RNE\ of test
    functions at each Metropolis iteration: (a) \RNE-based stopping rule; (b)
    $\Rbar=34$; (c) $\Rbar=89$; (d) $\Rbar=144$.} 
    \label{f:rne} 
\end{figure}

Figure \ref{f:rne} shows \RNE\ for all three test functions
$g_i(\theta,y_{1:t_\ell})$ at each Metropolis iteration of each cycle $\ell$.
The figure reports results for four different $M$ phase stopping rules: the
\RNE-based rule (with $E_1=0.35$, $E_2=0.90$ and $\Rmax=300$) and the
deterministic rule with $\Rbar$ = 34, 89 and 144.  The beginning of each $M$
phase is indicated by a vertical line.  The lower axis of each panel indicates
the Metropolis iteration (cumulating across cycles); the top axis indicates the
observation date $t$ at which each $M$ phase takes place; and the left axis
indicates \RNE.  For reference, Figure \ref{f:rne} includes a horizontal line
indicating $\RNE=0.35$.  This is the default target for the \RNE-based stopping
rule and serves as a convenient benchmark for the other rules as well.

In the early part of the sample, the deterministic rules execute many more
Metropolis steps than needed to achieve the nominal target of $\RNE=0.35$.
However, these require relatively little time because sample size $t$ is small. 
As noted above, there is little point in undertaking additional Metropolis
iterations once \RNE\ approaches one, as happens toward the beginning of the
sample for all three deterministic rules shown in the figure.

Toward the end of the sample, achieving any fixed \RNE\ target in the $M$ phase
requires more iterations due to the extreme non-Gaussianity of the posterior
(see Section \ref{ss:irregular} for a more detailed discussion of this issue).
The \RNE-based rule adapts automatically, performing iterations only as needed
to meet the \RNE\ target, implying more iterations as sample size increases in
this application.

At observation $t=474$, November 15, 1991, the $C$ phase terminates with very
low $\RSS$ (regardless of $M$ phase stopping rule used), the result of a return
that is highly unlikely conditional on the model and past history of returns.
The deterministic rules undertake additional Metropolis iterations to
compensate, as detailed in Section 3.3.  The \RNE-based rule also requires more
iterations than usual to meet the relevant threshold; but in this case the
number of additional iterations undertaken is determined algorithmically.

\begin{sidewaystable}
\caption{Moment approximations}
\label{t:moments}
\vspace{1ex}
\small
\centering
\begin{tabular}{cccccccccccccccccccc}
\hline\hline
     &         & \multicolumn{4}{c}{10 $\times$ Volatility State} & & \multicolumn{4}{c}{1000 $\times$ (3\% Loss Probability)} & & \multicolumn{4}{c}{10 $\times$ Skewness}\\
\cline{3-6}\cline{8-11}\cline{13-16}
$\Rbar$  & Compute &   E      & SD     & NSE    & RNE    &  & E      & SD     & NSE    & RNE    &  & E          & SD         & NSE        & RNE \\               
         &  Time   &          &        &        &        &  &        &        &        &        &  &            &            &            &     \\       
\hline                                                                                                                                                                
\multicolumn{16}{c}{\Vstrut{3ex}March 31, 2009}\\[1ex]
   5  &    91 &   -37.405 &    0.373 &    0.027 &    0.003 & &   95.877 &    7.254 &    0.513 &    0.003 & &   -1.999 &    0.823 &    0.063 &    0.003\\
%
   8  &   127 &   -37.388 &    0.390 &    0.029 &    0.003 & &   96.457 &    7.543 &    0.531 &    0.003 & &   -2.204 &    0.837 &    0.060 &    0.003\\
%
  13  &   205 &   -37.347 &    0.366 &    0.020 &    0.005 & &   97.311 &    7.167 &    0.361 &    0.006 & &   -2.348 &    0.774 &    0.034 &    0.008\\
%
  21  &   304 &   -37.382 &    0.352 &    0.014 &    0.009 & &   96.767 &    6.938 &    0.255 &    0.011 & &   -2.530 &    0.788 &    0.030 &    0.011\\
%
  34  &   482 &   -37.339 &    0.349 &    0.011 &    0.015 & &   97.582 &    6.913 &    0.211 &    0.016 & &   -2.569 &    0.812 &    0.021 &    0.023\\
%
  55  &   776 &   -37.340 &    0.353 &    0.007 &    0.034 & &   97.591 &    6.970 &    0.142 &    0.037 & &   -2.563 &    0.811 &    0.011 &    0.079\\
%
  89  &  1245 &   -37.330 &    0.364 &    0.006 &    0.055 & &   97.735 &    7.133 &    0.109 &    0.065 & &   -2.574 &    0.812 &    0.010 &    0.105\\
%
 144  &  2120 &   -37.332 &    0.355 &    0.004 &    0.141 & &   97.743 &    6.996 &    0.066 &    0.172 & &   -2.580 &    0.816 &    0.007 &    0.230\\
%
   *  &  2329 &   -37.334 &    0.359 &    0.003 &    0.170  & &   97.687 &    7.037 &    0.063 &    0.192 & &   -2.587 &    0.818 &    0.005 &    0.406\\ 
\multicolumn{16}{c}{\Vstrut{3ex}March 31, 2010}\\[1ex]
   5  &    91 &   -50.563 &    0.309 &    0.014 &    0.007 & &   0.596 &    0.253 &    0.013 &    0.006 & &  -2.078 &    0.816 &    0.063 &    0.003\\
%
   8  &   127 &   -50.514 &    0.309 &    0.011 &    0.012 & &   0.662 &    0.276 &    0.013 &    0.007 & &  -2.274 &    0.838 &    0.059 &    0.003\\
%
  13  &   205 &   -50.517 &    0.309 &    0.009 &    0.016 & &   0.696 &    0.284 &    0.010 &    0.013 & &  -2.425 &    0.771 &    0.033 &    0.008\\
%
  21  &   304 &   -50.512 &    0.310 &    0.007 &    0.031 & &   0.736 &    0.295 &    0.008 &    0.021 & &  -2.599 &    0.780 &    0.028 &    0.012\\
%
  34  &   482 &   -50.492 &    0.309 &    0.006 &    0.044 & &   0.754 &    0.301 &    0.006 &    0.045 & &  -2.616 &    0.795 &    0.018 &    0.030\\
%
  55  &   776 &   -50.491 &    0.310 &    0.004 &    0.097 & &   0.755 &    0.302 &    0.004 &    0.103 & &  -2.610 &    0.797 &    0.010 &    0.093\\
%
  89  &  1245 &   -50.482 &    0.315 &    0.003 &    0.164 & &   0.765 &    0.308 &    0.003 &    0.162 & &  -2.640 &    0.794 &    0.008 &    0.145\\
%
 144  &  2120 &   -50.483 &    0.314 &    0.002 &    0.381 & &   0.765 &    0.308 &    0.002 &    0.240 & &  -2.632 &    0.800 &    0.006 &    0.276\\
%
  *   &  2329 &   -50.482 &    0.315 &    0.002 &    0.460 & &   0.769 &    0.307 &    0.002 &    0.588 & &  -2.643 &    0.796 &    0.005 &    0.452\\
\hline\hline                                                                                                                                              
\end{tabular}                                                                                                                                                
\end{sidewaystable}

Table \ref{t:moments} reports details on the posterior mean approximations for
the two dates of interest. Total computation time for running the simulator
across the full sample is provided for reference.  The last line in each panel,
indicated by `*' in the $\Rbar$ field, uses the \RNE-based stopping rule.  For
both dates, the \NSE\ of the approximation declines substantially as $\Rbar$
increases. Comparison of the compute times reported in Table \ref{t:moments}
again suggests that increasing $\Rbar$ is more efficient for reducing \NSE\
than would be increasing $N$, up to at least $\Rbar=89$.

Some bias in the moment approximations is evident with low values of $\Rbar$.
The issue is most apparent for the 3\% loss probability on March 31, 2010.
Since volatility is low on that date, the probability of realizing a 3\% loss
is tiny and arises from tails of the posterior distribution of $\theta$, which
is poorly represented in small samples.  For example, with $\Rbar=13$, \RNE\ is
0.013, implying that each group of size $N=1024$ has an effective sample size
of only about 13 particles.  There is no evidence of bias for $\Rbar\geq 89$ or
with the \RNE-based rule.

Plots such as those shown in Figure \ref{f:rne} provide useful diagnostics and
are provided as a standard output of the software.  For example, it is easy to
see that with $\Rbar=34$ (panel (b) of the figure) not enough iterations are
performed in the $M$ phases, resulting in low \RNE\ toward the end of the
sample.  With lower values of $\Rbar$ the degradation in performance is yet
more dramatic.  As $\Rbar$ is increased to 89 and 144 in panels (c) and (d) of
the figure, respectively, the algorithm is better able to maintain particle
diversity through the entire sample.  The \RNE-based rule does a good job at
choosing an appropriate number of iterations in each $M$ phase and does so
without the need for user input or experimentation.

\subsection{Robustness to irregular posterior distributions}
\label{ss:irregular}

In the \egarch{23} model there are 2 permutations of the factors $v_{kt}$ and 6
permutations of the components of the normal mixture probability distribution
function of $\varepsilon _{t}$. This presents a severe challenge for
single-chain MCMC as discussed by Celeux et al.~(2000) and Jasra et al.~(2007),
and for similar reasons importance sampling is also problematic. The problem
can be mitigated (Fr\"uhwirth-Schnatter, 2001) or avoided entirely (Geweke,
2007) by exploiting the special \textquotedblleft mirror
image\textquotedblright\ structure of the posterior distribution. But these
models are still interesting as representatives of multimodal and ill-behaved
posterior distributions in the context of generic posterior simulators. We
focus here on the 6 permutations of the normal mixture in \egarch{23}.

Consider a $9\times 1$ parameter subvector $\psi$ with three distinct
values of the triplets $\left( p_{s},\mu_{s},\sigma_{s}\right) $ $\left(
s=A,B,C\right)$. There are six distinct ways in which these values could be
assigned to components $i=1,2,3$ of the normal mixture \eqref{normal_mixture}
in the \egarch{23} model. These permutations define six points $\psi_{u}$
$\left( u=1,\ldots ,6\right)$. For all sample sizes $t$, the posterior
densities $p\left( \psi_{u}\mid y_{1:t}\right)$ at these six points are
identical.  Let $\psi'$ be a different parameter vector with analogous
permutations $\psi'_u$ $(u=1,\dots,6)$.  As the sample adds evidence
$p\left( \psi_{u}\mid y_{1:t}\right) /p\left( \psi_{u^{\prime }}\mid
y_{1:t}\right) \overset{as}{\rightarrow }0$ or $p\left( \psi_{u}\mid
y_{1:t}\right) /p\left( \psi_{u^{\prime }}\mid y_{1:t}\right)
\overset{as}{\rightarrow }\infty $.  Thus, a specific triplet set of triplets
$\left( p_s,\mu_s ,\sigma_s \right)$ $(s=A,B,C)$ and its permutations will
emerge as pseudo-true values of the parameters (Geweke, 2005, Section 3.4).

\begin{figure}
\includegraphics[width=\textwidth]{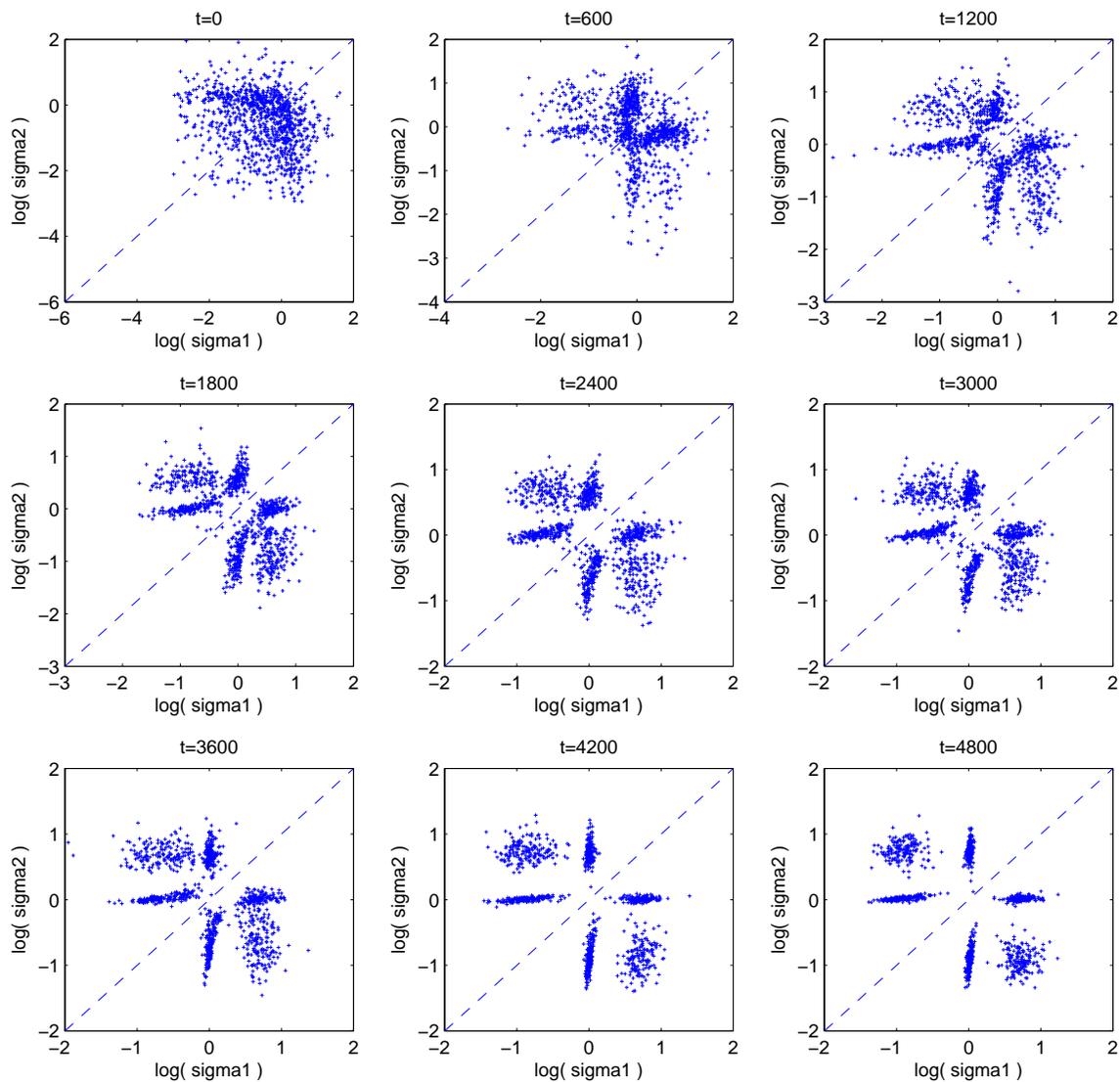}    
\caption{Scatterplots of a subset of
particles $\left(\log\sigma_{1},\log\sigma_{2}\right)$
from selected posterior distributions conditional on $y_{1:t}$ 
($J=64$, $N=4096$, $D_1=0.5$, $D_2=0.2$, $\Rbar=34$).}
\label{EGARCH_multimodal_sigma}
\end{figure}

The marginal distributions will exhibit these properties as well.  Consider the
pair $\left( \sigma _{1},\sigma _{2}\right) $, which is the case portrayed in
Figure \ref{EGARCH_multimodal_sigma}. The scatterplot is symmetric about the
axis $\sigma _{1}=\sigma _{2}$ in all cases. As sample size $t$ increases six
distinct and symmetric modes in the distribution gradually emerge. These
reflect the full marginal posterior distribution for the normal mixture
components of the \egarch{23} model (i.e., marginalizing on all other
parameters) that is roughly centered on the components $\left( p=0.17,\ \mu
=0.16,\ \sigma =0.40\right) $, $\left( p=0.85, \ \mu=0.01,\ \sigma =1.01\right)
$ and $\left( p=0.01,\ \mu =-1.36,\ \sigma =1.96\right) $. The progressive
decrease in entropy with increasing $t$ illustrates how the algorithm copes
with ill-behaved posterior distributions.  Particles gradually migrate toward
concentrations governed by the evidence in the sample. Unlike MCMC 
there is no need for particles to migrate between modes, and unlike
importance sampling there is no need to sample over regions eliminated by the
data (on the one hand) or to attempt to construct multimodal source
distributions (on the other). 

Similar phenomena are also evident for the other mixture parameters as well as
for the parameters of the GARCH factors.  The algorithm proposed in this paper
adapts to these situations without specific intervention on a case-by-case
basis.

\subsection{Comparison with Markov chain Monte Carlo}
\label{subsec:benchmarks}

To benchmark the performance of the algorithm against a more conventional
approach, we constructed a straightforward Metropolis random walk MCMC
algorithm, implemented in C code on a recent vintage CPU using the \egarch{23}
model. The variance matrix was tuned manually based on preliminary simulations,
which required several hours of investigator time and computing time. The
algorithm required 12,947 seconds for 500,000 iterations. The numerical
approximation of the moment of interest $1000\cdot \E \left[ P\left(
Y_{t+1}<-0.03\mid y_{1:t},\theta \right) \right] $ for March 31, 2010, the same
one addressed in Table \ref{t:moments}, produced the result $0.763$ and a \NSE\
of $0.003$. Taking the square of \NSE\ to be inversely proportional to the
number of MCMC iterations, an \NSE\ of $0.002$ (the result for the \RNE-based
stopping rule in Table \ref{t:moments}) would require about 1,840,000
iterations and 47,600 seconds computing time. Thus posterior simulation for the
full sample would require about 20 times as long using random walk Metropolis
in a conventional serial computing environment. 

Although straightforward, this analysis severely understates the advantage of
the sequential posterior simulator developed in this paper relative to
conventional approaches. The MCMC simulator does not traverse all six mirror
images of the posterior density. This fact greatly complicates attempts to
recover marginal likelihood from the MCMC simulator output; see Celeux et al.
(2000). To our knowledge the only reliable way to attack this problem using
MCMC is to compute predictive likelihoods from the prior distribution and the
posterior densities $p\left( \theta \mid y_{1:t}\right) $ $\left( t=1,\ldots
,T-1\right) $. Recalling that $T=5100$, and taking computation time to be
proportional to sample size (actually, it is somewhat greater) yields a time
requirement of $2550\times 28167=71,825,850$ seconds (2.28 CPU years), which is
almost 100,000 times as long as was required in the algorithm and
implementation used here.  Unless the function of interest $g(\theta,y_{1:T})$
is invariant to label switching (Geweke, 2007), the MCMC simulator must
traverse all the mirror images with nearly equal frequency.  As argued in
Celeux et al.~(2000), for all practical purposes this condition cannot be met
with any reliability even with a simulator executed for several CPU centuries.

This example shows that simple \textquotedblleft speedup
factors\textquotedblright\ may not be useful in quantifying the reduction in
computing time afforded by massively parallel computing environments.  For the
scope of models set in this work---that is, those satisfying Conditions
\ref{cond:prior_evaluate} through \ref{cond:prior_var}---sequential posterior
simulation is much faster than posterior simulation in conventional serial
computing environments. This is a lower bound. There are a number of routine
and reasonable objectives of posterior simulation, like those described in this
section, that simply cannot be achieved at all with serial computing but are
fast and effortless with SMC. Most important, sequential posterior simulation
in a massively parallel computing environment conserves the time, energy and
talents of the investigator for more substantive tasks.

\section{Conclusion}
\label{sec:conclusion}

Recent innovations in parallel computing hardware and associated software
provide the opportunity for dramatic increases in the speed and accuracy of
posterior simulation. Widely used MCMC simulators are not generically
well-suited to this environment, whereas alternative approaches like importance
sampling are.  The sequential posterior simulator developed here has attractive
properties in this context: inherent adaptability to new models; computational
efficiency relative to alternatives; accurate approximation of marginal and
predictive likelihoods; reliable and well-grounded measures of numerical
accuracy; robustness to irregular posteriors; and a well-developed theoretical
foundation.  Establishing these properties required a number of contributions
to the literature, summarized in Section 1 and then developed in Sections 2, 3
and 4.  Section 5 provided an application to a state-of-the-art model
illustrating the properties of the simulator.

The methods set forth in the paper reduce computing time dramatically in a
parallel computing environment that is well within the means of academic
investigators. Relevant comparisons with conventional serial computing
environments entail different algorithms, one for each environment.  Moreover,
the same approximation may be relatively more advantageous in one environment
than the other. This precludes generic conclusions about \textquotedblleft
speed-up factors.\textquotedblright\ In the example in Section
\ref{sec:application}, for a simple posterior moment, the algorithm detailed in
Section \ref{ss:details} was nearly 10 times faster than a competing
random-walk Metropolis simulator in a conventional serial computing
environment. For predictive likelihoods and marginal likelihood it was 100,000
times faster.  The parallel computations used a single graphics processing unit
(GPU). Up to eight GPU's can be added to a standard desktop computer at a cost
ranging from about \$350 (US) for a mid-range card  to about \$2000 (US) for a
high-performance Tesla card.  Computation time is inversely proportional to the
number of GPU's.

These contributions address initial items on the research agenda opened up by
the prospect of massively parallel computing environments for posterior
simulation. In the near term it should be possible to improve on the specific
contribution made here in the form of the algorithm detailed in Section
\ref{ss:details}.  Looking forward on the research agenda, a major component is
extending the class of models for which generic sequential posterior simulation
will prove practical and reliable. Large-scale hierarchical models,
longitudinal data, and conditional data density evaluations that must be
simulated all pose fertile ground for future work.

The research reported here has been guided by two paramount objectives. One is
to provide methods that are generic and minimize the demand for knowledge and
effort on the part of applied statisticians who, in turn, seek primarily to
develop and apply new models. The other is to provide methods that have a firm
methodological foundation whose relevance is borne out in subsequent
applications. It seems to us important that these objectives should be central
in the research agenda.

\small 
\pagebreak

\begin{center}
{\small {\Large References\medskip } }
\end{center}

\begin{description}


\item {\small Andrieu C, Doucet A, Holenstein A (2010). Particle Markov
chain Monte Carlo. Journal of the Royal Statistical Society, Series B 72:
1--33. }

\item {\small Baker JE (1985). Adaptive selection methods for genetic
algorithms. In Grefenstette J (ed.), Proceedings of the International
Conference on Genetic Algorithms and Their Applications, 101--111. Malwah
NJ: Erlbaum. }

\item {\small Baker JE (1987). Reducing bias and inefficiency in the
selection algorithm. In Grefenstette J (ed.) Genetic Algorithms and Their
Applications, 14--21. New York: Wiley. }

\item {\small Berkowitz, J (2001). Testing density forecasts with
applications to risk management. Journal of Business and Economic Statistics
19: 465--474. }



\item {\small Carpenter J, Clifford P, Fearnhead P (1999). Improved particle
filter for nonlinear problems. IEEE Proceedings -- Radar Sonar and
Navigation 146: 2--7. }

\item {\small Celeux G, Hurn M, Robert CP (2000). Computational and
inferential difficulties with mixture posterior distributions. Journal of
the American Statistical Association 95: 957--970. }


\item {\small Chopin N (2002). A sequential particle filter method for
static models. Biometrika 89: 539--551. }

\item {\small Chopin N (2004). Central limit theorem for sequential Monte
Carlo methods and its application to Bayesian inference. Annals of
Statistics 32: 2385--2411. }

\item Chopin N, Jacob P (2010).  Free energy sequential Monte Carlo,
application to mixture modelling. In: Bernardo JM, Bayarri MJ, Berger JO, Dawid
AP, Heckerman D, Smith AFM, West M (eds.), Bayesian Statistics 9. Oxford:
Oxford University Press.

\item {\small Chopin N, Jacob PI, Papaspiliopoulis O (2011). SMC$^{2}$: A
sequential Monte Carlo algorithm with particle Markov chain Monte Carlo
updates. Working paper.
\url{http://arxiv.org/PS_cache/arxiv/pdf/1101/1101.1528v2.pdf} }

\item Del Moral P, Doucet A, Jasra A (2006).  Sequential Monte Carlo samplers.
Journal of the Royal Statistical Society, Series B 68: 411--436.

\item Del Moral P, Doucet A, Jasra A (2011).  On adaptive resampling strategies
for sequential Monte Carlo methods.  Bernoulli, forthcoming.

\item {\small Diebold FX, Gunther TA, Tay AS. (1998). Evaluating density
forecasts with applications to financial risk management. International
Economic Review 39: 863--883. }


\item {\small Douc R, Moulines E (2008). Limit theorems for weighted samples
with applications to sequential Monte Carlo methods. The Annals of
Statistics 36: 2344--2376. }



\item {\small Durham G, Geweke J (2011).  Massively parallel sequential Monte Carlo for
Bayesian inference.  Working paper.
\url{http://www.censoc.uts.edu.au/pdfs/geweke_papers/gp_working_9.pdf}}

\item {\small Durham G, Geweke J (2013). Improving asset price prediction
when all models are false. Journal of Financial Econometrics, forthcoming.}

\item {\small Fearnhead P (1998). \emph{Sequential Monte Carlo methods in filter
    theory}. Ph.D.~thesis, Oxford University. }



\item {\small Flegal JM, Jones GL (2010).  Batch means and spectral variance estimates in 
Markov chain Monte Carlo.  Annals of Statistics 38: 1034--1070. }

\item {\small Fr\"uhwirth-Schnatter S (2001). Markov chain Monte Carlo
estimation of classical and dynamic switching and mixture models. Journal of
the American Statistical Association 96: 194--209. }

\item {\small Fulop A, Li J (2011). Robust and efficient learning: A
marginalized resample-move approach. Working paper.
\url{http://papers.ssrn.com/sol3/papers.cfm?abstract_id=1724203}}


\item {\small Gelman A, Roberts, GO, Gilks, WR (1996). Efficient Metropolis
jumping rules. In: Bernardo JM, Berger JO, Dawid AP,
Smith AFM (eds.), Bayesian Statistics 5. Oxford:
Oxford University Press. }

\item {\small Geweke J (1989). Bayesian inference in econometric models
using Monte Carlo integration. Econometrica 57: 1317--1340. }

\item {\small Geweke J (2005). Contemporary Bayesian Econometrics and
Statistics. Englewood Cliffs NJ: Wiley. }

\item {\small Geweke J (2007). Interpretation and inference in mixture
models: simple MCMC\ works. Computational Statistics and Data Analysis 51:
3529--3550. }

\item {\small Geweke J, Amisano G (2010). Comparing and evaluating Bayesian
predictive distributions of asset returns. International Journal of
Forecasting 26: 216--230. }






\item {\small Geweke J, Durham G, Xu H (2013).  Bayesian inference for logistic regression models
        using sequential posterior simulation.  
    \url{http://www.business.uts.edu.au/economics/staff/jgeweke/Geweke_Durham_XU_2013.pdf}}

\item {\small Gilks WR, Berzuini C (2001). Following a moving target --
Monte Carlo inference for dynamic Bayesian models. Journal of the Royal
Statistical Society, Series B 63: 127--146. }

\item {\small Gordon NG, Salmond DG, Smith AFM (1993). A novel approach to
non-linear and non-Gaussian Bayesian state estimation. IEEE Proceedings F:
Radar and Signal Processing 140: 107--113. }


\item {\small Hendeby G, Karlsson R, Gustafsson F (2010). Particle
filtering: The Need for Speed. EURASIP Journal on Advances in Signal
Processing doi:10.1155/2010/181403 }

\item {\small Herbst E, Schorfheide F (2012).  Sequential Monte Carlo Sampling for DSGE
    Models. Unpublished working paper.}


\item {\small Innovative Computing Laboratory (ICL), University of Tennessee
(2013). Matrix Algebra for GPU and Multicore Architectures.
\url{http://icl.cs.utk.edu/magma/} }


\item {\small Jasra A, Stephens DA, Holmes CC (2007). On population-based
simulation for static inference. Statistics and Computing 17: 263--279. }

\item {\small Kitagawa, G (1996). Monte Carlo filter and smoother for non-Gaussian 
        nonlinear state space models. Journal of Computational and Graphical
    Statistics 5: 1--25. }

\item {\small Kloek T, van Dijk HK (1978). Bayesian estimates of equation
system parameters: An application of integration by Monte Carlo.
Econometrica 46:1--19. }

\item {\small Kong A, Liu JS, Wong WH (1994). Sequential imputations and
Bayesian missing data problems. Journal of the American Statistical
Association 89: 278--288. }




\item {\small Lee L, Yau C, Giles MB, Doucet A, Homes CC (2010). On the
utility of graphics cards to perform massively parallel simulation of
advanced Monte Carlo Methods. Journal of Computational and Graphical
Statistics 19: 769--789. }



\item {\small Liu JS, Chen R (1995). Blind deconvolution via sequential
imputations. Journal of the American Statistical Association 90: 567--576. }

\item {\small Liu JS, Chen R (1998). Sequential Monte Carlo methods for
dynamic systems. Journal of the American Statistical Association 93:
1032--1044. }


\item {\small Matlab (2013). Matlab Parallel Computing Toolbox.
    \url{http://www.mathworks.com/products/parallel-computing/description5.html} }

\item {\small McGrayne SHB (2011). The Theory that Would Not Die: How Bayes'
Rule Cracked the Enigma Code, Hunted Down Russian Submarines, and Emerged
Triumphant from Two Centuries of Controversy. New Haven: Yale University
Press. }


\item {\small Nelson DB (1991). Conditional heteroskedasticity in asset
returns: A new approach. Econometrica 59: 347--370. }

\item {\small Nvidia (2013). Nvidia CUDA C Programming Guide, Version 5.0.
    \url{http://docs.nvidia.com/cuda/pdf/CUDA_C_Programming_Guide.pdf}}

\item {\small Rosenblatt, M (1952). Remarks on a multivariate
transformation. Annals of Mathematical Statistics 23: 470--472. }

             



\item {\small Smith, JQ (1985). Diagnostic checks of non-standard time
series models. Journal of Forecasting 4: 283--291. }

\item {\small Souchard MA, Wang Q, Chan C, Frelinger J, Cron A, West M
(2010). Understanding GPU programming for statistical computation: Studies
in Massively parallel massive mixtures. Journal of Computational Graphics
and Statistics 19: 419--438. }


\item {\small Tierney L (1994).  Markov chains for exploring posterior
distributions (with discussion and rejoinder). Annals of Statistics 22:
1701--1762.}


\end{description}

\end{document}